%% file: main.tex
\documentclass[11pt]{article}

\makeatletter
\@ifclassloaded{elsarticle}{%
    \newcommand{\ifelsarticle}{\expandafter\@firstoftwo}
}{%
    \newcommand{\ifelsarticle}{\expandafter\@secondoftwo}
}
\makeatother

\makeatletter
\ifelsarticle{%
    \usepackage[nofont, nofancyhdr, nonatbib]{prelude}
    \biboptions{square,numbers,sort&compress}
    \patchcmd{\section}{18\p@ \@plus 6\p@ \@minus 3\p@}{\bigskipamount}{}{}
    \patchcmd{\section}{9\p@ \@plus 6\p@ \@minus 3\p@}{\medskipamount}{}{}
    \patchcmd{\subsection}{12\p@ \@plus 6\p@ \@minus 3\p@}{\medskipamount}{}{}
    \patchcmd{\subsubsection}{12\p@ \@plus 6\p@ \@minus 3\p@}{\medskipamount}{}{}
    \patchcmd{\elsparagraph}{10\p@ \@plus 6\p@ \@minus 3\p@}{\medskipamount}{}{}
    \renewcommand{\subsubsection}{\paragraph}
}{%
    \usepackage[margin=1.25in]{geometry}
    \usepackage{prelude}
}
\makeatother

\newcommand{\testf}{Z}
\newcommand{\wk}{W}
\newcommand{\rema}{A_\res}
\newcommand{\agea}{A_\age}
\newcommand{\serv}{S}
\newcommand{\arri}{A}
\newcommand{\idle}{\mathrm{idle}}
\newcommand{\bs}{J}

\newcommand{\age}{\mathrm{age}}
\newcommand{\res}{\mathrm{res}}
\newcommand{\upti}{U}

\newcommand{\arate}{\lambda}
\newcommand{\rate}{\lambda}
\newcommand{\load}{\rho}
\newcommand{\rcy}{\mathrm{rcy}}
\newcommand{\rfresh}{r}
\newcommand{\rrcy}{{\r-\rcy}}

\newcommand{\acc}{\mathrm{acc}}
\newcommand{\racc}{{\r-\acc}}
\newcommand{\arv}{\mathrm{arv}}
\newcommand{\setup}{\mathrm{setup}}
\newcommand{\remamin}{A_{\min}}
\newcommand{\remamax}{A_{\max}}
\newcommand{\gittins}{\mathrm{Gtn}}

\newcommand{\indic}[1]{\mathbb{1}{\{#1\}}}
\newcommand{\rank}{\mathrm{rank}}
\newcommand{\excess}{\mathrm{e}}
\newcommand{\Sbig}{S_{\mathrm{big}}}

\newcommand{\abs}{\vert}

\newcommand{\loss}[1]{\ell_{\text{#1}}}

\newcommand{\mResidual}{m_\res}
\newcommand{\mRcy}{m_\rcy}
\newcommand{\mIdle}{m_{\mathrm{idle}}}
\newcommand{\mSetup}{m_\setup}

\makeatletter
\def\r-{\texorpdfstring{%
    \ifbool{mmode}{%
        r\mathhyphen%
    }{%
        \ifdefstring{\f@shape}{it}{\textrm{r}}{\textit{\textrm{\lowercase{r}}}}-%
    }%
}{r-}}
\makeatother

\def\M/#1/#2{\kendallNotation{M}{#1}{#2}}
\def\G/#1/#2{\kendallNotation{G}{#1}{#2}}
\makeatletter
\newcommand{\kendallNotation}[3]{%
    \ifbool{mmode}{%
        \text{#1/#2/}#3\@ifnextchar{/}{\@gobblesetup}{}
    }{%
        #1/#2/\texorpdfstring{%
            \ifbool{mmode}{#3}{%
                    \ifstrequal{#3}{k}{%
                    \ifdefstring{\f@shape}{it}{\textrm{k}}{%
                        \kern-0.04em\textup{\textit{\textrm{\lowercase{k}}}}%
                    }%
                }{#3}%
            }%
        }{#3}%
    }%
}
\def\@gobblesetup/setup{\text{/setup}}
\makeatother

\newcommand{\dtime}[1]{#1'}

\ifelsarticle{
    \title{Performance of the Gittins Policy in the \G/G/1 and \G/G/k, \\ With and Without Setup Times}

    \author[1]{Yige Hong}
    \ead{yigeh@cs.cmu.edu}
    \author[2]{Ziv Scully}
    \ead{zivscully@cornell.edu}
    \affiliation[1]{organization={Carnegie Mellon University, Computer Science Department}, addressline={5000 Forbes Ave}, city={Pittsburgh}, state={PA}, postcode={15213}, country={USA}}
    \affiliation[2]{organization={Cornell University, School of Operations Research and Information Engineering}, addressline={206 Rhodes Hall}, city={\\Ithaca}, state={NY}, postcode={14853-3801}, country={USA}}
}{
    \title[Gittins in the \G/G/1 and \G/G/k]{\vspace{-\baselineskip}Performance of the Gittins Policy in the \G/G/1 and \G/G/k, With and Without Setup Times}

    \author[Yige Hong and Ziv Scully]{%
        Yige Hong \\
        Carnegie Mellon University \\
        \texttt{yigeh@andrew.cmu.edu}
        \and
        Ziv Scully \\
        Cornell University \\
        \texttt{zivscully@cornell.edu}
    }
}

\date{January 2024}

\begin{document}

\ifelsarticle{}{\maketitle}
\begin{abstract}
    \input{abstract}
\end{abstract}
\ifelsarticle{\maketitle}{}

\input{intro}

\input{model-and-results}

\input{proofs}

\input{extensions-and-conclusion}

\section*{Acknowledgments}
This research was done in part while Ziv Scully was visiting the Simons Institute for Theoretical Computer Science at UC Berkeley, and in part while he was a FODSI postdoc at Harvard and MIT supported by NSF grant nos. DMS-2023528 and DMS-2022448. Yige Hong was supported by NSF grant no. ECCS-2145713.

\ifelsarticle{%
    \bibliographystyle{ACM-Reference-Format}%
}{%
    \bibliographystyle{ACM-ish}%
    \setlength{\bibsep}{\smallskipamount}
}
\bibliography{refs}

\appendix
\makeatletter
\ifelsarticle{%
    \renewcommand{\thesection}{\Alph{section}}%
    \renewcommand{\appendixname}{Appendix}%
    \def\@seccntformat#1{\appendixname~\csname the#1\endcsname.\quad}%
}{}
\makeatother

\input{appendix}

\end{document}

%% file: abstract.tex
How should we schedule jobs to minimize mean queue length? In the preemptive \M/G/1 queue, we know the optimal policy is the Gittins policy, which uses any available information about jobs' remaining service times to dynamically prioritize jobs. For models more complex than the \M/G/1, optimal scheduling is generally intractable. This leads us to ask: beyond the \M/G/1, does Gittins still perform well?

Recent results show Gittins performs well in the \M/G/k, meaning that its additive suboptimality gap is bounded by an expression which is negligible in heavy traffic. But allowing multiple servers is just one way to extend the \M/G/1, and most other extensions remain open. Does Gittins still perform well with non-Poisson arrival processes? Or if servers require setup times when transitioning from idle to busy?

In this paper, we give the first analysis of the Gittins policy that can handle any combination of (a)~multiple servers, (b)~non-Poisson arrivals, and (c)~setup times. Our results thus cover the \G/G/1 and \G/G/k, with and without setup times, bounding Gittins's suboptimality gap in each case. Each of (a), (b), and (c) adds a term to our bound, but all the terms are negligible in heavy traffic, thus implying Gittins's heavy-traffic optimality in all the systems we consider. Another consequence of our results is that Gittins is optimal in the \M/G/1 with setup times at all loads.

%%% Plain text version (no \, ~, or other LaTeX-isms) %%%

% How should we schedule jobs to minimize mean queue length? In the preemptive M/G/1 queue, we know the optimal policy is the Gittins policy, which uses any available information about jobs' remaining service times to dynamically prioritize jobs. For models more complex than the M/G/1, optimal scheduling is generally intractable. This leads us to ask: beyond the M/G/1, does Gittins still perform well?
%
% Recent results show Gittins performs well in the M/G/k, meaning that its additive suboptimality gap is bounded by an expression which is negligible in heavy traffic. But allowing multiple servers is just one way to extend the M/G/1, and most other extensions remain open. Does Gittins still perform well with non-Poisson arrival processes? Or if servers require setup times when transitioning from idle to busy?
%
% In this paper, we give the first analysis of the Gittins policy that can handle any combination of (a) multiple servers, (b) non-Poisson arrivals, and (c) setup times. Our results thus cover the G/G/1 and G/G/k, with and without setup times, bounding Gittins's suboptimality gap in each case. Each of (a), (b), and (c) adds a term to our bound, but all the terms are negligible in heavy traffic, thus implying Gittins's heavy-traffic optimality in all the systems we consider. Another consequence of our results is that Gittins is optimal in the M/G/1 with setup times at all loads.

%% file: intro.tex
\section{Introduction}
\label{sec:intro}

We consider the classic problem of preemptively scheduling jobs in a queue to minimize mean number-in-system, or equivalently mean response time (a.k.a. sojourn time). Even in single-server queueing models, this can be a nontrivial problem whose answer depends on the information available to the scheduler. The simplest case is when the scheduler knows each job's size (a.k.a. service time), for which the optimal policy is Shortest Remaining Processing Time (SRPT) \citep{schrage_proof_1968}: always serve the job of least remaining work.

In the more realistic case of scheduling with unknown or partially known job sizes, the optimal policy is only known for the \M/G/1. It is called the \emph{Gittins} policy (a.k.a. Gittins index policy) \citep{aalto_gittins_2009, aalto_properties_2011, gittins_multi-armed_2011, scully_gittins_2021}. Based on whatever service time information is available for each job, Gittins assigns each job a scalar \emph{rank} (i.e. priority), then serves the job of least rank. For example, SRPT is the special case of Gittins where job sizes are known exactly, and a job's rank is its remaining work. More generally, a job's rank is, roughly speaking, an estimate of its remaining work based on whatever information is available.

The Gittins policy is known to be optimal in the \M/G/1 \citep{gittins_multi-armed_2011, scully_gittins_2021}. But plenty of systems and models have more complex features, including:
\label{sec:intro:feature_list}
\*[(a)] \emph{Multiple servers}, such as the \M/G/k.
\* \emph{Non-Poisson arrival processes}, such as the \G/G/1 (more specifically, the GI/GI/1).
\* \emph{Periods of server unavailability}, such as models with setup times.
\*/
Either (a) or (b) alone makes optimal scheduling intractable. Combining all three, as in the \G/G/k with setup times (\G/G/k/setup), only adds to the challenge.

With optimality out of reach, we are left to find a tractable near-optimal policy. We thus ask:
\begin{quote}
    How well does Gittins perform in systems with features (a), (b), and~(c) like the \G/G/k/setup?
\end{quote}
Gittins is a natural candidate because its definition naturally generalizes beyond the \M/G/1, even if its optimality proof does not \citep{gittins_multi-armed_2011}. For instance, in a \G/G/k, Gittins simply serves the $k$~jobs of $k$ least ranks, or all jobs if there are fewer than~$k$.

Only feature (a) has been addressed in full generality in prior work \citep{scully_gittins_2020, grosof_optimal_2022, scully_new_2022}. Specifically, it is known that in the \M/G/k, the additive suboptimality gap of Gittins (abbreviated in equations to ``$\gittins$'') is bounded by \citep{scully_new_2022}\footnote{%
    Throughout this paper, $\log$ is the natural logarithm.}
\[
    \label{eq:intro:mgk_bound}
    \E{N}_{\M/G/k}^\gittins - \inf_{\text{policies }\pi} \E{N}_{\M/G/k}^\pi 
    \leq C (k-1) \log \frac{1}{1-\load}.
\]
Let us briefly explain the notation used in~\cref{eq:intro:mgk_bound}:
\* $\E{N}_{\M/G/k}^\pi$ is the mean number-in-system under policy~$\pi$ in \M/G/k.
\* $k$ is the number of servers.
\* $\load \in [0, 1)$ is the \emph{load} (a.k.a. utilization), namely the average fraction of servers that are busy.
\* $C \approx 3.775$ is a constant.
\*/
A notable feature of \cref{eq:intro:mgk_bound} is that under mild conditions \citep{scully_gittins_2020}, the right-hand side is dominated by $\inf_\pi \E{N}_{\M/G/k}^\pi$, the performance of the optimal policy, in the \emph{heavy-traffic limit}, meaning as $\load \to 1$. That is, as the \M/G/k gets busier and busier, the difference between Gittins's performance and that of the optimal policy becomes negligible. Gittins is thus considered \emph{heavy-traffic optimal} in the \M/G/k.

The above progress on analyzing Gittins in the multiserver \M/G/k is certainly promising for handling~(a). But, as we explain in more detail in \cref{sec:intro:obstacles}, key steps of the prior \M/G/k analysis rely on Poisson arrivals and uninterrupted server availability, so they cannot handle (b) and~(c).

\subsection{Results: Performance Bounds and Heavy-Traffic Optimality}
\label{sec:intro:results}

We give the first analysis of the Gittins policy for systems with any combination of (a)~multiple servers, (b)~G/G arrivals, and (c)~setup times. We frame our results in terms of the \G/G/k with setup times (\G/G/k/setup), which can have all three features. But they also apply to systems with a subset of the features, such as the \G/G/1 or \M/G/k/setup, because we can view them as special cases of the \G/G/k/setup.

Our main results, presented in full in \cref{sec:results}, can be summarized as
\[
    \label{eq:intro:results}
    \E{N}_{\G/G/k/setup}^\gittins - \inf_{\text{policies }\pi} \E{N}_{\G/G/k/setup}^\pi
    \leq \loss{(a)} + \loss{(b)} + \loss{(a)\&(c)},
\]
where each term on the right-hand side is a ``suboptimality loss'' that is non-zero when the model has the features in the subscript. For example, the \M/G/k/setup has only features~(a) and~(c), so $\loss{(b)} = 0$. A particularly notable case is the \M/G/1/setup, which has only feature~(c), so all three loss terms in \cref{eq:intro:results} are zero. Indeed, we show that Gittins is optimal among non-idling policies in the \M/G/1/setup, a previously unknown result.

Our result generalizes prior work on Gittins in the \M/G/k in the sense that $\loss{(a)}$ turns out to be the right-hand side of \cref{eq:intro:mgk_bound}. Remarkably, the other loss terms, $\loss{(b)}$ and~$\loss{(a)\&(c)}$, are uniformly bounded at all loads. This implies that, under mild conditions, Gittins is heavy-traffic optimal in the \G/G/k/setup.

\subsubsection*{Stability of the \G/G/k/setup under Complex Scheduling Policies}

Perhaps a more basic question than optimizing mean response time is stability: under which scheduling policies is the \G/G/k/setup stable? This is an open question which is outside the scope of our work. As such, our results hold under an assumption on the stability region of the \G/G/k/setup (\cref{assump:stability}). We believe the assumption always holds, giving a partial proof sketch in \cref{sec:stability_sketch}.

\subsubsection*{Beyond the \G/G/k/setup}

The techniques underlying our results are very general, applying even beyond the \G/G/k/setup. In \cref{sec:extensions}, we sketch how our results could be extended to other systems.
\* Building on the theme of multiple servers, we consider systems with \emph{multiserver jobs}, which must be simultaneously served by multiple servers \citep{brill_queues_1984}. Due to the prevalence of multiserver jobs in cloud computing, these models have received lots of recent attention \citep{grosof_optimal_2022, harchol-balter_open_2021, wang_zero_2021, hong_sharp_2022, rumyantsev_stability_2017, rumyantsev_three-level_2022, grosof_wcfs_2022}.
\* Building on the theme of non-Poisson arrivals, we consider \emph{batch arrivals} of jobs \citep{chaudhry_first_1983}.
\* Building on the theme of setup times, we consider \emph{generalized vacations}, which model a variety of scenarios where servers are temporarily unavailable \citep{fuhrmann_note_1984, fuhrmann_stochastic_1985}.
\*/

\subsection{Main Obstacles and Key Ideas}
\label{sec:intro:obstacles}

While there is a substantial literature on scheduling in the \M/G/1 \citep[Part~VII]{harchol-balter_performance_2013}, much less is known as soon as we introduce features (a), (b), and~(c). Any two of these, let alone all three, yields a system where optimal scheduling has never been studied. This is perhaps unsurprising in light of the fact analyzing these systems under First-Come First-Served (FCFS) is already very difficult. See \citet{li_simple_2017} and references therein for a review of the \G/G/k, and similarly for \citet{williams_mmk_2022} for the \M/G/k/setup. Even in the \G/G/1, we only know the optimal scheduling policy for known job sizes, when it is SRPT \citep{schrage_proof_1968}.

Fortunately, recent advances analyzing Gittins in the \M/G/k \citep{scully_gittins_2020, scully_new_2022, grosof_optimal_2022} give us hope in the form of a new avenue for analyzing performance. \Citet{scully_gittins_2020} introduce a new queueing identity, now known as \emph{WINE} \citep{scully_new_2022} (\cref{sec:background}),\footnote{%
    \Citet[Section~2.2.3]{scully_new_2022} notes that WINE builds upon several similar identities that precede it \citep{glazebrook_parallel_2001, glazebrook_analysis_2003, righter_extremal_1990, banerjee_heavy_2022}.}
which relates the number of jobs~$N$ in the system to, roughly speaking, the amount of \emph{work} in the system. This is helpful because bounding the amount of work in an \M/G/k, which WINE turns into a bound on $\E{N}$, turns out to be significantly easier than directly bounding $\E{N}$.

WINE holds in any queueing system, including the \G/G/k/setup, so we can and do use the same overall strategy of bounding work, then using WINE to turn the work bound into an $\E{N}$ bound. However, there are significant obstacles to carrying out this strategy in the \G/G/k/setup.

\subsubsection*{Non-Poisson Arrivals}

The first step of our strategy is to analyze the amount of work in the system. The approach that prior work takes to analyze the \M/G/k is to use a \emph{work decomposition law}. This is a result which, in its most general form, relates the amount of work in a generic system with M/G arrivals to the amount of work in a ``resource-pooled'' \M/G/1 experiencing the same arrivals. The prior \M/G/k bound in \cref{eq:intro:mgk_bound} comes from the fact that the \M/G/k and resource-pooled \M/G/1 turn out to have similar amounts of work.
We would like to take a similar approach with the \G/G/k/setup. Unfortunately, the combination of G/G arrivals and multiple servers rules out using existing work decomposition laws (\cref{sec:prior:work_decomposition}).

To overcome this, we prove a \emph{new work decomposition law for G/G arrivals} (\cref{sec:work-decomp}). We view this as the main technical contribution that makes our results possible. Indeed, by combining WINE and our new work decomposition law, the $\loss{(a)}$ and $\loss{(b)}$ terms of \cref{eq:intro:results} follows relatively easily. But the $\loss{(c)}$ term and heavy-traffic analysis present additional obstacles, as discussed below.

\subsubsection*{Setup Times}

One of the key observations behind the prior \M/G/k analysis is that whenever there are $k$ jobs in the system, all servers are occupied. This implies that in terms of work, an \M/G/k never falls too far behind an \M/G/1, where the \M/G/1 experiences the same arrivals and has the same total service capacity. But in an \M/G/k/setup or \G/G/k/setup, there is no analogous limit to how far behind an \M/G/1/setup or \G/G/k/setup system can be, because there is no limit on the number of jobs that might arrive during a setup time.

To overcome this, we perform a novel analysis of setup times to bound the number of arrivals during one setup time \emph{in expectation}. This analysis is the basis of the $\loss{(a)\&(c)}$ term of~\cref{eq:intro:results}.

\subsubsection*{Heavy Traffic Analysis}

The above ideas are enough to prove the bound in~\cref{eq:intro:results}. But the question remains: is the right-hand side of \cref{eq:intro:results} small or large relative to $\inf_\pi \E{N}_{\G/G/k/setup}^\pi$, the performance of the optimal policy? If the latter dominates the former in the $\load \to 1$ limit, then Gittins is optimal in heavy traffic. The right-hand side grows as $O\gp[\big]{\log \frac{1}{1 - \load}}$, so the main challenge is to give a lower bound on the performance of the optimal policy. In prior work on the \M/G/k, one can use SRPT in a resource-pooled \M/G/1 as a lower bound on the optimal policy, which is helpful because the SRPT has been analyzed in heavy traffic \citep{lin_heavy-traffic_2011}. We would like to use the same approach with the \G/G/1 as the lower bound, but SRPT has never been analyzed in the heavy-traffic \G/G/1.

To overcome this, we give the \emph{first heavy-traffic analysis of SRPT in the \G/G/1}. This provides a lower bound on $\inf_\pi \E{N}_{\G/G/k/setup}^\pi$, which turns out to be enough for our purposes. The key ingredient of our heavy-traffic analysis is, once again, our new work decomposition law for G/G arrivals, underscoring its importance as our key technical contribution.

\subsection{Outline}

The rest of the paper is organized as follows:
\* \Cref{sec:prior} reviews related work.
\* \Cref{sec:model} describes our \G/G/k/setup model, and in particular details of the setup times.
\* \Cref{sec:results} presents our main results on Gittins: suboptimality gap bounds (\cref{thm:main_multiserver, thm:main_single-server}) and heavy-traffic optimality (\cref{thm:heavy_traffic}).
\* \Cref{sec:proof-overview} gives a high-level overview of how we prove our main results.
\* \Cref{sec:background} reviews necessary background on Gittins and WINE.
\* \Cref{sec:work-decomp} proves a \emph{new work decomposition law for systems with G/G arrivals}. This is the key technical contribution that underlies all of our other results.
\* \Cref{sec:proofs} proves the suboptimality gap bounds (\cref{thm:main_multiserver, thm:main_single-server}).
\* \Cref{sec:heavy-traffic} proves heavy-traffic optimality (\cref{thm:heavy_traffic}). The key step involves giving \emph{first heavy-traffic analysis of SRPT in the \G/G/1}, a result of independent interest.
\*/

\section{Related Work}
\label{sec:prior}

\subsection{Optimal Scheduling in Queues}
\label{sec:prior:scheduling}

\subsubsection*{Gittins in Single-Server Systems}

The Gittins policy was originally conceived to solve the Markovian multi-armed bandit problem \citep{gittins_multi-armed_2011, gittins_dynamic_1974}, but it was soon adapted to also solve the problem of scheduling in an \M/G/1 to minimize mean number of jobs and similar metrics. See \citet{scully_gittins_2021} and the references therein for a review of Gittins in the \M/G/1. However, aside from some particular cases \citep{schrage_proof_1968, righter_scheduling_1989}, the degree to which Gittins performs well in the \G/G/1 or \G/G/1/setup was previously unknown.

The ``SOAP'' technique of \citet{scully_soap_2018} can be used to analyze the performance of the Gittins policy in the \M/G/1. However, while SOAP is convenient for analyzing any fixed size distribution (e.g. numerically), using it to prove theorems that hold for all size distributions is cumbersome \citep[Section~1.1]{scully_simple_2020}. Moreover, SOAP is limited to the \M/G/1 and, thanks to an extension by \citet{vreumingen_queueing_2019}, the \M/G/1/setup. Analyzing Gittins with G/G arrivals or multiple servers seems to be beyond SOAP \citep[Appendix~B]{scully_optimal_2021}.

\subsubsection*{Gittins in Multi-Server Systems}

Gittins is known to be suboptimal with multiple servers \citep{gittins_multi-armed_2011}, but researchers have studied the extent to which the suboptimality gap is large or small. The earliest results of this type analyzed an \M/M/k with Bernoulli feedback \citep{glazebrook_parallel_2001} and \emph{nonpreemptive} \M/G/k with Bernoulli feedback \citep{glazebrook_analysis_2003}. These results proved (in the latter case, under an additional assumption) constant suboptimality gaps for Gittins in these systems. But both models are somewhat restrictive, excluding, for instance, heavy-tailed job size distributions that are common in computer systems \citep{harchol-balter_exploiting_1997, harchol-balter_network_1996, crovella_self-similarity_1997, park_self-similar_2000, peterson_data_1996}. More recent work, which we discussed in \cref{sec:intro}, overcomes these limitations to bound the performance of Gittins in the \M/G/k for general job sizes, including heavy-tailed sizes \citep{scully_characterizing_2020, grosof_optimal_2022, scully_new_2022}. However, all of the above work assumes M/G arrivals with no server unavailability.

\subsection{Setup Times}

\subsubsection*{Multiserver Models}

A significant line of previous work has studied the \M/M/k/setup with exponential setup times and FCFS scheduling \citep{artalejo_analysis_2005, gandhi_mgk_2009, gandhi_server_2010, gandhi_mgk_2013, gandhi_exact_2014, pender_law_2016}. Among those works, \citet{gandhi_mgk_2009} and \citet{gandhi_mgk_2013}
also demonstrate that their results generalize to \M/G/k/setup with exponential setup times via simulation or analyzing special examples. Recently, \citet{williams_mmk_2022} go beyond exponential setup times, studying \M/M/k/setup with deterministic setup times and FCFS scheduling. However, none of these prior works apply to general setup times, non-Poisson arrivals, or scheduling policies beyond FCFS.

We note that \citet{glazebrook_analysis_2003}, who studies the nonpreemptive \M/G/k with Bernoulli feedback, actually studies a more general model that allows for certain types of server unavailability, such as server breakdowns. However, setup times are not covered by \citet{glazebrook_analysis_2003}. It is likely that more general future work could simultaneously cover setup times, server breakdowns, and other types of server unavailability. See \cref{sec:extensions:generalized_vacations} for additional discussion.

\subsubsection*{Single-Server Models}

Compared with multiserver models, single-server models with setup times are better understood \citep{welch_generalized_1964, doshi_note_1985, bischof_analysis_2001, choudhury_batch_1998, he_flow_1995, li_new_1994, miyazawa_decomposition_1994, fuhrmann_stochastic_1985, fuhrmann_note_1984, doshi_queueing_1986, upadhyaya_queueing_2016}. See \citet{doshi_queueing_1986} for a survey of the work before 1986 and \citep{upadhyaya_queueing_2016} for a more recent survey. These works consider various arrival and service processes, as well as other types of server unavailability in addition to setup times.

However, none of the above works discuss optimal scheduling in the presence of setup times. Progress was made by \citet{vreumingen_queueing_2019}, who obtains the mean response time of Gittins in the \M/G/1/setup as a special case of a more general analysis (\cref{sec:prior:scheduling}). But the analysis does not show that other policies might outperform Gittins, nor does it apply to the \G/G/1/setup.

\subsection{Decomposition Laws in Queues}
\label{sec:prior:work_decomposition}

There is a long tradition of proving work decomposition laws for queueing systems \citep{fuhrmann_note_1984, fuhrmann_stochastic_1985, boxma_pseudo-conservation_1987, miyazawa_decomposition_1994, glazebrook_parallel_2001, glazebrook_analysis_2003, scully_gittins_2020, scully_new_2022, doshi_generalizations_1990}. Most of these laws take the form
\[
    \E{\text{work in complex system with M/G arrivals}}
    = \E{\text{work in M/G/1}} + \E{\text{extra work due to complexity}}.
\]
For example, if the complex system is an M/G/1/setup, the extra work from complexity depends on the setup time distribution. Most work decomposition laws are actually even stronger, holding \emph{distributionally} instead of just in expectation.

We need a work decomposition law where the complexity includes, among other factors, having multiple servers. Such a result for M/G arrivals is relatively recent \citep{scully_gittins_2020, scully_new_2022}, and no such result exists for G/G arrivals. While there are work decomposition laws for G/G arrivals in the literature \citep{doshi_note_1985, doshi_generalizations_1990, miyazawa_decomposition_1994}, to the best of our knowledge, they apply only to single-server models with vacations. To the best of our knowledge, we prove the first work decomposition law for G/G arrivals that holds for multiserver systems like the \G/G/k.

%% file: model-and-results.tex
\section{Model}
\label{sec:model}

\subsection{Core Queueing Models: \G/G/k, \G/G/1, \M/G/k, and \M/G/1}
\label{sec:model:ggk}

We consider a \G/G/k queueing model with a single central queue and $k$ identical servers. The system experiences \emph{G/G arrivals}: jobs arrive one-by-one with i.i.d. \textit{interarrival times}, and each job has an i.i.d. \emph{size}, or service requirement. Interarrival times and job sizes are independent of each other. We denote a generic random interarrival time by~$\arri$ and a generic random job size by~$S$.\footnote{%
    This arrival process is often referred to more specifically as GI/GI arrivals, with the ``I'' emphasizing the independence assumption. Under this convention, G/G arrivals include even more general stationary arrival processes where independence does not hold. In this work, we focus only on the independent case, so we write simply ``G/G'' instead of ``GI/GI'' for brevity.}

At any moment of time, a job in the system can be served by one server. Any jobs not in service wait in the queue. Once a job's service is finished, it departs. We follow the convention that each of the $k$ servers has service rate~$1/k$. A job of size~$S$ thus requires $kS$ time in service to finish. This convention gives all systems we study the same maximum total service rate, namely $k \cdot 1/k = 1$, and thereby the same stability condition.

The name ``\G/G/k'' denotes the fact that the system has G/G arrivals and~$k$ servers. When $A$ is exponentially distributed, we write \emph{M/G} in place of G/G, as in ``\M/G/k''.

\subsubsection*{Scheduling Policies}

The scheduling policy decides, at every moment in time, which job is in service at which server. We consider a preempt-resume model where preemption occurs without delay or loss of work.

The scheduling objective is minimizing the \emph{mean number of jobs} in the system. We denote the mean number of jobs in system~SYS under scheduling policy~$\pi$ by $\E{N}_{\text{SYS}}^\pi$, omitting the ``SYS'' and/or ``$\pi$'' if there is no ambiguity. By Little's law \citep{little_littles_2011}, minimizing $\E{N}$ is equivalent to minimizing \emph{mean response time}, the average amount of time a job spends between its arrival and departure.

We use a flexible model of how much the scheduler knows about each job's size (\cref{sec:model:scheduling}). We restrict attention to \emph{non-idling} policies, which are those that never unnecessarily leave servers idle. Nevertheless, our results have implications even for idling policies (\cref{sec:results:remarks}).

As a consequence of frequent preemption, the server can share one server between multiple jobs. We formalize this in \cref{sec:model_extra:scheduler_actions}, but our presentation does not depend on the formal details.

\subsubsection*{Load and Stability}

We write $\lambda = 1/\E{\arri}$ for the average arrival rate and $\load = \lambda \E{S}$ for the system's \emph{load}, or utilization. One can think of $\load$ as the average fraction of servers that are busy. It is clear that $\load < 1$ is a necessary condition for stability (unless both $A$ and $S$ are deterministic), so we assume this throughout.

Some of our results are stated for the \emph{heavy-traffic limit}. For our purposes, this limit, denoted $\load \to 1$, refers to a limit as the job size distribution~$S$ remains constant, and the interarrival time distribution~$A$ is scaled uniformly down with its mean approaching the mean job size. That is, the system with load~$\load$ has interarrival time $A_\rho = A_1 / \rho$ for some fixed distribution~$A_1$, where $\E{A_1} = \E{S}$.

It seems intuitive that $\load < 1$ should be sufficient for stability under non-idling policies, and it is in the \G/G/1 \citep{loynes_stability_1962}. But to the best of knowledge, there are no results characterizing stability of the \G/G/k under complex scheduling policies. Even under FCFS, proving stability of the \G/G/k is not simple, because the system can be stable even when it never empties \citep{kiefer_theory_1955, whitt_existence_1982, sigman_one-dependent_1990, morozov_stability_2021}. Setup times further complicate the matter.

We consider the question of proving stability of the \G/G/k/setup under arbitrary non-idling scheduling policies to be outside the scope of this paper, so we simply assume stability when $\load < 1$. We expect this is indeed the case, giving the initial steps of a proof sketch in \cref{sec:stability_sketch}.

\begin{assumption}
    \label{assump:stability}
    For all $\load < 1$, the \G/G/k/setup is stable under all non-idling scheduling policies, including the Gittins policy.
\end{assumption}

\subsubsection*{Additional Assumption on Interarrival Times}

Our results for G/G arrivals depend on ``how non-Poisson'' arrival times are, which we quantify using the following assumption.

\begin{assumption}
    \label{assump:remabd}
    There exist $\remamin, \remamax \in \R_{\geq 0}$ such that $\E{\arri - a \given \arri > a} \in [\remamin, \remamax]$ for all $a \geq 0$. 
    That is, letting the \emph{interarrival age}~$\agea$ be the time since the last arrival and \emph{residual interarrival time}~$\rema$ be the amount of time until the next arrival, we have
    \[
        \E{\rema \given \agea} \in [\remamin, \remamax] \quad \text{with probability~$1$.}
    \]
    One may always use $\remamin = \inf_{a\geq 0} \E{\arri - a \given \arri > a}$ and $\remamax = \sup_{a\geq 0} \E{\arri - a \given \arri > a}$, so this assumption boils down to the latter being finite.
\end{assumption}

Our results use \cref{assump:remabd} via the quantity $\lambda(\remamax - \remamin)$, which we can think of as measuring ``how non-Poisson'' arrival times are. In the Poisson case, one may use $\remamin = \remamax = 1 / \lambda$, so $\lambda(\remamax - \remamin) = 0$.

Many interarrival distributions~$A$ satisfy \cref{assump:remabd}, such as all phase-type distributions. One can also think of \cref{assump:remabd} as a relaxation of the well-known \emph{New Better than Used in Expectation} (NBUE) property, which is the special case where $\remamax = \E{A}$. The main distributions ruled out by \cref{assump:remabd} are various classes of heavy-tailed distributions, e.g. power-law tails.

\subsection{Setup Times}
\label{sec:model:setup_times}

In addition to the basic \G/G/k model defined above, we also consider models in which servers require \emph{setup times} to transition from idle to busy. We denote these models with an extra ``/setup'', as in \G/G/k/setup. Whenever a server switches from idle to busy, it must first complete an i.i.d. amount of \emph{setup work}, denoted~$\upti$. Like work from jobs, servers complete setup work at rate~$1/k$, so setup \emph{work} $\upti$ results in setup \emph{time} $k \upti$. Setup work amounts are independent of interarrival times and job sizes.

For the purposes of stating our results and proofs in a unified manner, we consider the \G/G/k without setup times to be the special case of the \G/G/k/setup where $\upti = 0$ with probability~$1$.

In our model, a server can be in one of three states:
\* \emph{Setting up}, i.e. doing setup work.
\* \emph{Busy}, i.e. serving a job.
\* \emph{Idle}, i.e. neither serving a job nor doing setup work.
\*/
In the \G/G/1/setup, state transitions are straightforward: the server goes
from setting up to busy when it finishes its setup work,
from busy to idle when no jobs remain in the system, and
from idle to setting up when a job arrives to an empty system.
But in the \G/G/k/setup, the transitions are more complicated. This is because there are several design choices to make, and thus multiple models that might be studied. For example, if we already have one busy server, how many jobs should there be in the queue before we start setting up a second server? For concreteness, we study one particular setup time model, described below, but our work still has implications for alternative models (\cref{sec:results:remarks, sec:extensions:generalized_vacations}).

In the \G/G/k/setup, we use the following setup time model: a server transitions
\* from setting up to busy when it finishes its setup work,
\* from busy to idle when the system has fewer jobs than busy servers, and
\* from idle to setting up when the system has fewer busy or setting up servers than jobs.
\*/
Thus, transitions to setting up are triggered by arrivals, and transitions to idle are generally triggered by departures. Servers transition ``one at a time'', e.g. an arrival triggers at most server to start setting up.

Note that once a setup time begins, it is never canceled, even if the job whose arrival triggered the setup time begins service at another server. Unless another job arrives during the setup time, the server will transition from setting up to busy, then immediately back to idle. Not canceling setup times is a natural modeling choice for some systems, e.g. computer systems where cutting power during startup is undesirable. But our techniques could also be used to analyze setup times that can be canceled (\cref{sec:extensions:generalized_vacations}).

\subsection{What the Scheduler Knows About Jobs' Sizes}
\label{sec:model:scheduling}

We consider a flexible model of the scheduler's knowledge called the \emph{Markov-process job model} \citep{scully_gittins_2020, scully_gittins_2021, scully_new_2022}. In this model, each job has a \emph{state}, which inhabits some \emph{job state space}~$\mathbb{X}$, representing what the scheduler knows about that job. Each job's state evolves as an i.i.d. absorbing continuous-time Markov-process $\{X(t)\}_{t \geq 0}$ on some state space~$\mathbb{X}$, where $X(t)$ is the state of the job after it has received $t \geq 0$ service. That is, a job's state evolves while it is in service but stays static while it is in the queue. There is an extra absorbing state $\top \not\in \mathbb{X}$, corresponding to the job finishing, i.e. jobs exit the system when their state becomes~$\top$.

We call $\{X(t)\}_{t \geq 0}$ the \emph{job Markov process}. We can recover the job size from the job Markov process as
\[
    S = \inf\curlgp{t \geq 0 \given X(t) = \top}.
\]
To clarify, the amount of service~$t$ in $X(t)$ is measured in \emph{work} rather than \emph{time}, so jobs evolve at rate $1/k$ when served in a $k$-server system (\cref{sec:model:ggk}). As discussed in \cref{sec:model_extra:technicalities}, we make some purely technical assumptions on the job Markov process (e.g. r.c.l.l.) to ensure Gittins is well defined.

We assume that the scheduler always knows the state of all jobs in the system, which we denote by $(X_1, \dots, X_N)$. We also assume the scheduler knows the dynamics of the job Markov process, e.g. the size distribution~$S$. A job's state thus encodes everything the scheduler knows about the job. For example, given a job in state~$x$, the scheduler knows its \emph{remaining work}, namely the amount of service the job needs to complete, is distributed as\footnote{%
    Abusing notation slightly, we interpret conditioning $X(0) = x$ as the usual notion of starting the job Markov process from state~$x$. This gets around the corner case where the initial state $X(0)$ is never~$x$.}
\[
    S(x) = \gp[\big]{\inf\curlgp{t \geq 0 \given X(t) = \top} \given X(0) = x}.
\]

Below are two concrete examples of the Markov-process job model. These are extremes: the first is perfect size information, and the other is zero size information beyond knowing the distribution~$S$. For additional examples, including cases where the scheduler has partial size information, see \citet[Section~3]{scully_gittins_2020}.

\begin{example}
    \label{ex:known_model}
    The case of \emph{known sizes} is when a job's state is its remaining work. The state space is $\mathbb{X} = (0, \infty)$, the initial state is distributed as $X(0) \sim S$, and the absorbing state is $\top = 0$. During service, the job's state decreases at rate~$1$. That is, $X(t) = (X(0) - t)^+$. In state~$x$, the remaining work is $S(x) = x$.
\end{example}

\begin{example}
    \label{ex:unknown_model}
    The case of \emph{unknown sizes} is when a job's state is the amount of service it has received so far. The state space is $\mathbb{X} = [0, \infty)$, the initial state is $X(0) = 0$, and the absorbing state is an isolated point~$\top$. During service, the job's state increases at rate~$1$ and has a chance to jump to~$\top$, with the exact chance depending on the distribution of~$S$. That is, $X(t) = t$ until the job completes, after which $X(t) = \top$. In state~$x$, the remaining work is the conditional distribution $S(x) \sim (S - t \given S > t)$.
\end{example}

\subsection{The Gittins Policy}
\label{sec:model:gittins}

The scheduling policy we focus on in this work is the \emph{Gittins} policy (a.k.a. \emph{Gittins index} policy). Gittins is primarily known for the fact that it minimizes $\E{N}$ in the \M/G/1 \citep{gittins_multi-armed_2011, scully_gittins_2021}. In formulas, we abbreviate Gittins to ``$\gittins$'', as in $\E{N}_{\G/G/k/setup}^\gittins$.

The Gittins policy has a relatively simple form. It assigns each job a numerical priority, called a \emph{rank}, where lower rank is better. Gittins always serves the job or jobs of least rank,\footnote{%
    Much literature on the Gittins policy uses the opposite convention, where higher numbers are better. These works typically call a job's priority its \emph{index} \citep{gittins_multi-armed_2011, aalto_gittins_2009, aalto_properties_2011}, which is the reciprocal of its rank \citep{scully_gittins_2021}.}
and it is non-idling, serving as many jobs as the number of available servers allows. Gittins determines ranks using a \emph{rank function}
\[
    \rank_\gittins : \mathbb{X} \to \R_{\geq 0},
\]
assigning $\rank_\gittins(x)$ to a job in state $x \in \mathbb{X}$. A job's rank thus depends only on its own state.

It turns out that our proofs do not directly use the definition of Gittins's rank function. 
As such, we specify the Gittins rank function for the concrete job Markov processes in \cref{ex:known_model, ex:unknown_model}, the latter of which in particular explains the key intuition. We refer the curious reader to \cref{sec:model_extra:gittins_rank} for the general definition, though we emphasize it does not play a direct role in our proofs.

\begin{example}
    \label{ex:known_gittins}
    In the case of \emph{known} job sizes, it turns out that Gittins reduces to SRPT, which always serves the job of least remaining work. A job's rank is thus its remaining work. Recalling from \cref{ex:known_model} that a job's state is its remaining work under known sizes, we simply have $\rank_\gittins(x) = x$.
\end{example}

\begin{example}
    \label{ex:unknown_gittins}
    In the case of \emph{unknown} job sizes, recall from \cref{ex:unknown_model} that a job's state~$x$ is the amount of service it has already received. In this case, the Gittins rank function is \citep{gittins_multi-armed_2011}
    \[
        \rank_\gittins(x) = \inf_{y>x} \frac{\E{\min\{S, y\}-x \given S>x}}{\P{S\leq y \given S>x}}.
    \]
    The intuition for this formula is as follows. Consider a job in state~$x$, and suppose we start serving the job, but decide to ``give up'' if it reaches state~$y$. On the right-hand side, the numerator is the expected amount of service until we either complete the job or give up, and the denominator is the probability the job completes before we give up. The right-hand side is thus a ``service-per-completion'' ratio, giving an expected amount of effort it would take to finish one job in expectation. A job's rank under Gittins is the best service-per-completion ratio one can obtain by optimally choosing the state~$y$ in which to give up.
\end{example}

\section{Main Results}
\label{sec:results}

We now state our main results. All of our results hold under the assumptions of \cref{sec:model}, and in particular \cref{assump:stability, assump:remabd}. As in \cref{sec:intro}, we can view a \G/G/k/setup system, or any special case thereof, by whether it has (a)~multiple servers, (b)~non-Poisson arrivals, and (c)~setup times. Our bounds use the quantities
\[
    \loss{(a)} &= C(k - 1)\log\frac{1}{1 - \load},
    \\
    \loss{(b)} &= \lambda(\remamax - \remamin),
    \\
    \loss{(c)} &= \1(\P{U > 0} > 0) \gp[\big]{2(k - 1) + \arate(\remamax + k\E{\upti_\excess})},
\]
where $C = \frac{9}{8\log 1.5} + 1 \approx 3.775$. The idea is that $\loss{(a)}$ is the loss due to feature~(a), as it is nonzero only for systems with $k \geq 2$ servers, and similarly for $\loss{(b)}$ and $\loss{(c)}$.\footnote{%
    The reason \cref{eq:intro:results} has an $\loss{(a)\&(c)}$ term instead of an $\loss{(c)}$ term is because it summarizes both \cref{thm:main_multiserver, thm:main_single-server}.}

\begin{theorem}\label{thm:main_multiserver}
    The performance gap between the Gittins policy in \G/G/k/setup and the optimal policy in \G/G/1 is bounded by
    \[
        \E{N}_{\G/G/k/setup}^\gittins - \inf_\pi \E{N}_{\G/G/1}^\pi
        \leq \loss{(a)} + \loss{(b)} + \loss{(c)}.
    \]
\end{theorem}

Note that although \cref{thm:main_multiserver} is not directly about the suboptimality gap of Gittins policy in $\G/G/k/setup$, it still provides an upper bound on the suboptimality gap, because the optimal performance of \G/G/1 is a lower bound to \G/G/k/setup. This is because servers in \G/G/k/setup have speed~$1/k$ (\cref{sec:model:ggk}), so the \G/G/1 can mimic any policy in the \G/G/k/setup through processor sharing and idling.

With that said, in the special case of the non-idling \G/G/1/setup, we can prove a stronger result that drops the $\loss{(c)}$ term by comparing to a \G/G/1/setup instead of a \G/G/1.

\begin{theorem}\label{thm:main_single-server}
    In the \G/G/1/setup, the performance gap between the Gittins policy and the optimal non-idling policy is bounded by
    \[
        \E{N}_{\G/G/1/setup}^\gittins - \inf_{\pi} \E{N}_{\G/G/1/setup}^\pi
        \leq \loss{(b)}.
    \]
    In particular, in the \M/G/1/setup, the Gittins policy minimizes $\E{N}$ among non-idling policies.
\end{theorem}

The suboptimality gap in \cref{thm:main_multiserver} is constant when $k = 1$ and $O\gp[\big]{\log \frac{1}{1 - \load}}$ when $k \geq 2$. In both cases, the gap grows more slowly in the $\load \to 1$ limit than $\E{N}_{\G/G/1}^\pi$, implying heavy-traffic optimality.

\begin{theorem}
    \label{thm:heavy_traffic}
    In the \G/G/k/setup, if either $k = 1$ or $\E{S^2 (\log S)^+} < \infty$, and if either $S$ or~$A$ is not deterministic, the Gittins policy is heavy-traffic optimal. Specifically, $\lim_{\rho\to 1} {\E{N}_{\G/G/k/setup}^\gittins}/{\inf_\pi \E{N}_{\G/G/1}^\pi} = 1$. 
\end{theorem}

We prove this result in \cref{sec:heavy-traffic}. The main obstacle is showing a lower bound on $\E{N}_{\G/G/1^\pi}$. We use SRPT as a lower bound, so the first step of the proof is to analyze SRPT in the heavy-traffic \G/G/1 (\cref{lem:srpt-gg1-vs-mg1}). We find its performance is within a constant factor of SRPT in the heavy-traffic \M/G/1.

\subsection{Remarks on Main Results}
\label{sec:results:remarks}

\subsubsection*{Alternative Setup Time Models}

Because \cref{thm:main_multiserver, thm:heavy_traffic} compare Gittins in the \G/G/k/setup to the optimal policy in a \G/G/1, it also effectively compares Gittins under our setup time model to Gittins in the \G/G/k with essentially any other setup time model. This is because the \G/G/1 serves as a lower bound for alternative setup time models, not just our specific \G/G/k/setup. The takeaway is that changing the setup time model does not significantly impact performance in heavy traffic, which makes intuitive sense: servers seldom set up if they are usually busy.

\subsubsection*{Idling Policies}

We say in \cref{sec:model:ggk} that we only consider non-idling policies, but \cref{thm:main_multiserver} still compares Gittins to idling policies. This is because the optimal policy in a \G/G/1 is clearly non-idling, and, by the discussion above, it gives a lower bound on any policy, idling or non-idling, in the \G/G/k/setup.

Why, then, does \cref{thm:main_single-server} only compare to the optimal non-idling policy? This is because in the \G/G/1/setup, idling the server can change when setup times occur. By idling with one job in the queue, one can effectively control when setup times occur by choosing when to start the job, without waiting for an arrival. Hypothetically, this could improve performance in the \G/G/1/setup. That is, idling effectively allows a policy to use an alternative setup time model, so we rule it out.

\subsubsection*{Opportunities for a Tighter Bound}

The bound shown in \cref{thm:main_multiserver} represents a trade-off between proving a tight bound and stating the result simply. We prioritized making the statement as simple as possible while ensuring the $\loss{(a)}$ term matches the bound for Gittins in the \M/G/k from prior work \citep{scully_new_2022}. But there are at least two clear avenues for tightening our bounds.

First, there are other bounds on Gittins in the \M/G/k \citep{scully_gittins_2020, grosof_optimal_2022}, which can be better than $\loss{(a)}$ in some cases. We believe that one may take $\loss{(a)}$ to be the minimum of these bounds, but doing so would complicate the result and proof without substantially changing the main takeaway.

Second, our bound is loose in light traffic. We should have $\E{N} \to 0$ as $\load \to 0$, but our $\loss{(b)}$ and $\loss{(c)}$ terms remain nonzero at all loads. One can sharpen our analysis for the special case of the \M/G/k/setup to obtain a suboptimality gap that becomes zero in the $\load \to 0$ limit. But we doubt even this improved bound is very tight at low loads, so we omit the extra casework.

%% file: proofs.tex
\section{Proof Overview}\label{sec:proof-overview}

In this section, we give an overview of the proofs of our main results: bounds on Gittins's suboptimality gap (\cref{thm:main_multiserver,thm:main_single-server}) and Gittins's heavy-traffic optimality (\cref{thm:heavy_traffic}). At a high level, our proofs work by combining two queueing identities: \emph{WINE}, which is from prior work; and a novel \emph{work decomposition law}, which is built on similar decomposition results from prior work (\cref{sec:prior:work_decomposition}).

The first tool, WINE (\cref{lem:wine}), expresses the mean number-in-system in terms of \emph{mean \r-work}~$\E{W_r}$ \citep{scully_new_2022, scully_gittins_2020, scully_gittins_2021}:
\[
    \E{N} = \int_0^\infty \frac{\E{W_r}}{r^2} \d{r}.
\]
A system's \r-work $W_r$ is the total service required to serve all jobs in the system until they all either complete or reach rank greater than~$r$, as determined by $\rank_\gittins$ (\cref{sec:model:scheduling}). For example, $\infty$-work is the total remaining work of all jobs, which we call \emph{total work} or simply \emph{work}. See \cref{sec:background} for details.

The second tool, the work decomposition law (\cref{lem:work-decomp}), implies bounds on $\E{W_r}$ under any policy, including Gittins. Combining this with WINE yields bounds on~$\E{N}$. Our proof thus boils down to three steps:
\* Proving the work decomposition law (\cref{sec:proof-overview:work-decomp}).
\* Using the work decomposition law to bound Gittins's suboptimality gap (\cref{sec:proof-overview:suboptimality-gap}).
\* Using the suboptimality gap bounds to show Gittins is heavy-traffic optimal (\cref{sec:proof-overview:heavy-traffic}).
\*/

\subsection{New Tool: Work Decomposition Law for G/G Arrivals}
\label{sec:proof-overview:work-decomp}

Our work decomposition law, \cref{lem:work-decomp}, characterizes mean \r-work $\E{W_r}$ in the \G/G/k/setup. For simplicity of presentation, below we focus on the special case of mean total work~$\E{W}$.

\Cref{lem:work-decomp} implies that in the \G/G/k/setup under any policy~$\pi$,
\[
    \E{W}^\pi - \E{W}_{\G/G/1} \leq \frac{\E{J_\idle W}^\pi}{1-\load} + \frac{\E{J_\setup W}^\pi}{1-\load} + \load (\remamax - \remamin). \label{eq:work-decomp-for-total-work}
\]
Above, $\E{W}_{\G/G/1}$ is the mean work in a non-idling G/G/1, which is policy-invariant; and $J_\idle$ and $J_\setup$ are the fraction of idle and setting-up servers, respectively. Flipping the sign on the $\load (\remamax - \remamin)$ term yields a lower bound instead of an upper bound.

The work decomposition law decomposes work $\E{W}^\pi$ into the policy-invariant term $\E{W}_{\G/G/1}$, plus error terms that can depend on the policy~$\pi$. Each error term characterizes the consequence of a complicating factor that \G/G/k/setup has on the top of the \G/G/1 system:
\*[(a)] The first term is due to having multiple servers. It vanishes when $k=1$, as then $J_\idle=0$ if $W>0$.\footnote{%
    When generalizing \cref{eq:work-decomp-for-total-work} from total work to \r-work, there are actually two terms due to having multiple servers. But both vanish when $k=1$.}
\* The second term is due to the setup time. It vanishes if servers do not need setup, as then $J_\setup = 0$.
\* The third term is due to non-Poisson arrivals. It vanishes for Poisson arrivals, as then $\remamax = \remamin$.
\*/

\subsubsection*{How We Prove the Work Decomposition Law}

The proof of work decomposition laws in prior work involves viewing $W$ as a process in the steady state and analyzing its continuous changes and jumps. 
This strategy works well in M/G systems, because all times have an equal chance of seeing $W$ jump up due to an arrival.
But in G/G systems, the chance of having an arrival in the next moment depends on $\agea$, the amount of time since the previous arrival. The jumps of $W$ are thus more complicated to analyze.

The key idea in our proof is to smooth out the non-constant jumping rate of $W$.
Specifically, we consider the process $W - \rho \rema$, which only differs from $W$ by one interarrival time.
This process decreases at a constant rate of $1-\rho$. When an arrival happens, the process jumps, but the expected change is $\E{S} - \rho \E{A} = 0$. Therefore, arrivals only have a ``second-order'' effect on~$W$, which makes them easier to analyze. This idea builds upon similar smoothing approaches in recent queueing literature \citep{braverman_heavy_2017, miyazawa_diffusion_2015}.

\subsection{From Work Decomposition to Suboptimality Gap Bounds}
\label{sec:proof-overview:suboptimality-gap}

We focus here on proving \cref{thm:main_multiserver}, commenting only briefly on the similar proof of \cref{thm:main_single-server}.

Combining our work decomposition law with WINE gives a formula for Gittins's suboptimality gap that has the same types of error terms as \eqref{eq:work-decomp-for-total-work}. Each error term in the work decomposition law will result in one term in the suboptimality gap $\ell_{(a)} + \ell_{(b)} + \ell_{(c)}$ in \cref{thm:main_multiserver}, after doing the integration and applying some additional treatments that are specific to each term.

Among the three error terms, $\ell_{(a)}$ can be derived similarly to prior work on the \M/G/k \citep{scully_gittins_2020, scully_new_2022, grosof_optimal_2022}, and $\ell_{(b)}$ follows from \cref{assump:remabd}. But the term corresponding to setup, $\ell_{(c)}$, requires a new analysis.
We demonstrate the intuition by bounding $\E{J_\setup W}$ in \cref{eq:work-decomp-for-total-work}. First, we write $\E{J_\setup W} = \frac{1}{k} \sum_{i = 1}^k \E{J_{\setup,i} W}$, where $J_{\setup, i} = \1(\text{server $i$ is setting up})$. Observe that 
\[
    \E{J_{\setup,i} W} = \P{J_{\setup,i}=1} \, \E{W \given J_{\setup,i}=1}.
\]
Intuitively, $\P{J_{\setup,i}=1}$ should be diminishing as the load gets heavy because the queue length will get longer the server~$i$ will be turned off less frequently. The second factor, $\E{W \given J_{\setup,i}=1}$, should be bounded because given that the server $i$ is setting up, the work in the system should be no more than the work that arrives during the setup, plus the work that already exists when the setup happens.

For the proof of \cref{thm:main_single-server}, which gives a tighter bound for the single-server case, we apply WINE and work decomposition law in the same way as above. We will get an expression for $\E{N}^\gittins_{\G/G/1/setup}$ in terms of one $\G/G/1$ term, and two error terms corresponding to non-Poisson arrivals and setup times. Instead of analyzing the setup term as in the proof of \cref{thm:main_multiserver}, we make the simple observation that the setup term is the same for all non-idling policies, so it does not contribute to the suboptimality gap.

\subsection{From Suboptimality Gap Bounds to Heavy-Traffic Optimality}
\label{sec:proof-overview:heavy-traffic}

\Cref{thm:main_multiserver} provides an upper bound on the suboptimality gap of Gittins policy in \G/G/k/setup. To show that the suboptimality gap is small compared with $\inf_\pi \E{N}_{\G/G/k/setup}^\pi$ and establish heavy-traffic optimality of the Gittins policy, we need a lower bound on $\inf_\pi \E{N}_{\G/G/k/setup}^\pi$. This lower bound can be obtained by analyzing $\E{N}_{\G/G/1}^{\text{SRPT}}$ because SRPT gives the optimal number-in-system in \G/G/1 with known job sizes \citep{schrage_proof_1968}, which is no more than the optimal number-in-system in \G/G/k/setup achievable by a policy that does not know the job size. 

We use WINE and work-decomposition law, in a similar way as in the proofs in the suboptimality gaps, to connect SRPT's performance in the \G/G/1 to its performance in the \M/G/1. Our end result (\cref{lem:srpt-gg1-vs-mg1}) shows that $\E{N}_{\G/G/1}^{\text{SRPT}}$ is a constant factor away from $\E{N}_{\M/G/1}^{\text{SRPT}}$ as $\rho\to 1$. This lets us to use the known heavy-traffic asymptotics of SRPT in the \M/G/1 \citep{lin_heavy-traffic_2011} to lower bound $\E{N}_{\G/G/1}^{\text{SRPT}}$ and thus show Gittins's heavy-traffic optimality in the \G/G/k/setup.

\section{Background on WINE and \r-Work}\label{sec:background}

A queueing system's \emph{work}~$\wk$ is the total remaining work of all jobs in the system: $\wk = \sum_{i = 1}^N S(X_i)$, where $S(X_i)$ is the remaining work of job~$i$ (\cref{sec:model:scheduling}). We define \r-work similarly: we first define the remaining \r-work of a job, then define the system's \r-work to be the sum of all jobs' remaining \r-work.

\begin{definition}
    \label{def:r-work}
    Let $r \geq 0$. The \emph{remaining \r-work} of a job in state~$x$, denoted $S_r(x)$, is the amount of service it needs until it either finishes or reaches a state whose rank is at least~$r$:
    \[
        S_r(x)
        &= \text{amount of service a job starting at~$x$ needs to finish or reach rank at least~$r$} \\
        &= \gp[\big]{\inf\curlgp{t \geq 0 \given X(t) = \top \mathrel{\text{or}} \rank_\gittins(X(t)) \geq r} \given X(0) = x}. 
    \]
    A system's \emph{\r-work}, denoted~$\wk_r$, is the sum of the remaining \r-work of all jobs in the system: $\wk_r = \sum_{i = 1}^N S_r(X_i)$. 
\end{definition}

We now present the WINE identity. It holds for any scheduling policy that has access to only the current and past system states (\cref{sec:model:scheduling}). For concreteness, we state WINE for our specific queueing model, but it holds in essentially any system which uses the Markov-process job model.

\begin{lemma}[WINE {\citep[Theorem~6.3]{scully_gittins_2020}}]\label{lem:wine} In the \G/G/k/setup under any scheduling policy,
    \[
         N &= \int_0^\infty \frac{\E{\wk_r \given X_1, \dots, X_N}}{r^2} \d{r},
         & %\\
         \E{N} &= \int_0^\infty \frac{\E{\wk_r}}{r^2} \d{r}.
    \]
\end{lemma}

WINE, which integrates the entire system's \r-work to get the number of jobs, follows from a more basic identity, sometimes called ``single-job WINE'' \citep{scully_new_2022}, which integrates a single job's remaining \r-work.

\begin{lemma}[Single-Job WINE {\citep[Lemma~6.2]{scully_gittins_2020}}]\label{lem:single-job-wine}
    For any job state $x \in \mathbb{X}$, we have $\int_0^\infty \frac{\E{S_r(x)}}{r^2} \d{r} = 1$. 
\end{lemma}

One subtlety about WINE is that while it applies to any scheduling policy, the definition of \r-work uses Gittins's rank function. As a general rule, this makes analyzing Gittins's performance using WINE easier than analyzing other policies' performance using WINE, particularly when proving upper bounds, though there are some exceptions \citep{scully_uniform_2022, scully_new_2022}. Our work is no exception: we prove our main results by upper bounding Gittins's \r-work and lower-bounding the optimal policies' \r-work.

\subsection{Additional Definitions for Reasoning About \r-Work}
\label{sec:background:r-stuff}

WINE reduces the problem of analyzing the steady-state mean number of jobs $\E{N}$ to the problem of analyzing steady-state mean \r-work $\E{W_r}$. In order to analyze \r-work, we need to understand the means by which the amount of \r-work in the system changes over time. This section introduces the standard concepts and vocabulary used to discuss \r-work \citep{scully_new_2022, scully_gittins_2020, scully_soap_2018}.

The definitions in this section are parameterized by a rank $r \geq 0$, as denoted by a prefix ``\r-''. We often drop this prefix when the rank~$r$ is clear from context or not important to the discussion.

\subsubsection*{Relevant, Irrelevant, Fresh, and Recycled Jobs}

We call a job \emph{\r-relevant} whenever its rank is less than~$r$. Otherwise, the job is \emph{\r-irrelevant}.

Whether a job is \r-relevant or \r-irrelevant varies over time. Consider one job's journey through the system. When the job arrives, it may be either, depending on its initial state~$X(0)$. As the job is served, its rank can go up and down, so it may alternate between \r-relevant and \r-irrelevant, possibly multiple times, before eventually finishing and exiting the system.

From the above discussion, it is evident that there are two ways for the amount of \r-work in a system to increase. Both are important, so we introduce terminology for discussing both.

\begin{definition}\label{def:r-fresh}
    We call an \r-relevant job \emph{\r-fresh} if it has been \r-relevant ever since its arrival. That is, new arrivals that are initially \r-relevant are \r-fresh until they either finish or become \r-irrelevant.
    \* We write $S_\rfresh = S_r(X(0))$ for the random amount of service during which a newly arrived job is \r-fresh. Arriving jobs may be \r-irrelevant, so it may be that $S_\rfresh = 0$ with nonzero probability.
    \* We call $\load_\rfresh = \lambda \E{S_\rfresh}$ the \emph{\r-fresh load}. It is the average rate \r-work is added by new arrivals.
    \*/
\end{definition}

\begin{definition}
    \label{def:r-recycled}
    We call an \r-relevant job \emph{\r-recycled} if it was \r-irrelevant at some point in the past. We refer to the moment a job switches from \r-irrelevant to \r-recycled as an \emph{\r-recycling}.
    \* We write $\lambda_\rrcy$ for the average rate of \r-recyclings. 
    \* We write $S_\rrcy$ for the random amount of \r-work added by a single \r-recycling.
    \* We call $\load_\rrcy = \lambda_\rrcy \E{S_\rrcy}$ the \emph{\r-recycled load}. It is the average rate \r-work is added by \r-recyclings.
    \*/
\end{definition}

Note that \r-recyclings can only occur when an \r-irrelevant job is in service. This is because a job's rank depends only on its state, which only changes during service (\cref{sec:model:scheduling, sec:model:gittins}).

We assume for ease of presentation that $\lambda_\rrcy < \infty$ for all $r \geq 0$. This is actually not always the case, but the assumption can be straightforwardly relaxed by appealing to more advanced Palm calculus techniques. See \citet[Appendix~E]{scully_gittins_2020} for details.

\subsubsection*{Server States}

Finally, we need notations for discussing how many servers are doing \r-work or setting up.

\begin{definition}\label{def:J}
    \leavevmode
    \*[beginpenalty=10000] We write $J_r$ for the fraction of servers that are serving \r-relevant jobs.\footnote{%
        If processor sharing is occurring, we interpret $J_r$ as the total service rate assigned to \r-relevant jobs (\cref{sec:model_extra:scheduler_actions}).}
    Note that $\E{J_r} = \load_\rfresh + \load_\rrcy$.
    \* We write $J_\setup$ for the fraction of servers that are setting up.
    \*/
\end{definition}

\section{Work Decomposition Law for G/G Arrivals}\label{sec:work-decomp}

In this section, we introduce a new formula for the expected \r-work in systems with G/G arrivals, which we call the \emph{work decomposition law}. The formula decomposes mean \r-work into multiple terms, where each term is either fixed or easy to analyze.
To present the formula and its proof, we need to first review a concept called \emph{Palm expectation}. We give an intuitive review that suffices for our purposes below, referring the reader to \citet{miyazawa_rate_1994} and \citet{baccelli_elements_2003} for more formal treatments.

\subsection{Palm Expectation and Notation}
\label{sec:work-decomp:palm}

Consider the amount of work in the system as seen by arriving jobs. Intuitively speaking, to define the long-run average as seen by arrivals, we would like to start with the steady-state average, then condition on the event ``an arrival is occurring''. However, the probability that an arrival occurs at any given point in time is~$0$, so we cannot naively apply the usual conditional expectation formula. Palm expectations solve this problem, giving a way of formally defining a notion that corresponds to conditioning on events like ``an arrival is occurring''.
We use two main Palm expectations in our analysis of \r-work:
\* The \emph{Palm expectation at arrivals}, denoted~$\E_\arv{\cdot}$.
\* The \emph{Palm expectation at \r-recyclings}, denoted~$\E_\rrcy{\cdot}$.
\*/
That is, $\E_\arv{V}$ is the expectation of random quantity~$V$, which may depend on the system state, sampled at the time ``just before'' an arrival occurs; and $\E_\rrcy{V}$ is the same, but for \r-recyclings instead of arrivals.

Although we define $\E_\arv{\cdot}$ and $\E_\rrcy{\cdot}$ as referring to the moment ``just before'' an event, it is also helpful to be able to refer to the quantity ``just after'' the event. To facilitate this, we use the following notation:
\* Within $\E_\arv{\cdot}$, we denote the remaining \r-work of the arriving job by~$S_\rfresh$.
\** By our independence assumptions (\cref{sec:model:ggk}), $S_\rfresh$ is independent of the system state.
\* Within $\E_\rrcy{\cdot}$, we denote the remaining \r-work of the job being recycled by~$S_\rrcy$.
\** In general, $S_\rrcy$ is \emph{not} independent of the system state, because recyclings are caused by events happening within the system.
\*/

\subsubsection*{A Notation Shortcut}
In addition to the Palm expectation, we also define the notation $\E_\racc{\cdot}$ as
\[
    \label{eq:r-acc}
    \E_\racc{V} = \frac{\E{(1-J_r) V} + \rate_\rrcy \E_\rrcy{\serv_\rrcy V}}{1-\load_\rfresh},
\]
where $J_r$ is the fraction of servers that are busy with \r-relevant jobs (\cref{sec:background:r-stuff}). One can interpret $\E_\racc{\cdot}$ as a type of Palm expectation. 
For instance, it behaves like an expectation in the sense that $\E_\racc{v} = v$ for deterministic values~$v$. But for our purposes, it suffices to understand $\E_\racc{\cdot}$ as simply a notation shortcut.

\subsubsection*{Excess Distributions}

We introduce a piece of notation that occurs frequently in queueing and renewal theory \citep{asmussen_applied_2003, harchol-balter_performance_2013}, including in the statement of our work decomposition law below.

\begin{definition}
    \label{def:excess}
    Given a nonnegative distribution~$V$, we define its \emph{excess}, denoted~$V_\excess$, to be the distribution with tail\footnote{%
    As a corner case, if $V = 0$ with probability~$1$, we let $V_\excess = 0$ with probability~$1$.}
    \[
        \P{V_\excess > t} &= \int_t^\infty \frac{\P{V > u}}{\E{V}} \d{u},
        &
        &\text{which has mean}
        &
        \E{V_\excess} &= \frac{\E{V^2}}{2\E{V}}.
    \]
\end{definition}

\subsection{Statement and Proof of Work Decomposition Law}

Now we are ready to present the work decomposition law for systems with G/G arrivals. We state the result for \r-work~$W_r$, but we can apply the result to total work~$W$ by taking an $r \to \infty$ limit.

\begin{restatable:theorem}[Work Decomposition Law for G/G Arrivals]\label{lem:work-decomp}
    In the \G/G/k/setup under any policy~$\pi$,
    \[
        \E{\wk_r}^\pi &= \frac{\load_\rfresh \gp*{\E{\gp*{S_\rfresh}_\excess} -  \E{S_\rfresh} + \E*{\arri_\excess}} + \load_\rrcy \E{\gp{\serv_\rrcy}_\excess}}{1-\load_\rfresh} + \E_\racc*{\wk_r}^\pi - \load_\rfresh \E_\racc{\rema}^\pi.
    \]
\end{restatable:theorem}

\begin{proof}
We drop the superscript $\pi$ throughout. For each $r$, define
\begin{equation}\label{eq:testf-def}
    \testf_r = \frac{1}{2}\gp*{\wk_r - \load_\rfresh \rema}^2.
\end{equation}
We use $\testf_r$ as a ``test function'' and extract information about the G/G system by looking at how $\testf_r$ changes and applying Miyazawa's Rate Conservation Law (RCL) \citep{miyazawa_rate_1994}. We discuss why $\testf_r$ is the right choice of test function in \cref{remark:rho-r-x-a-res} below.

Over time, the quantity $\testf_r$ changes in the following ways: continuous change as \r-work and remaining arrival time decrease over time, jump when a new job arrives, and jump at an \r-recycling event. 
We use $\dtime \testf_r$ to denote the continuous change of $\testf_r$, use $\Delta_\arv \testf_r$ to denote the jumps of $\testf_r$ at arrival times, and use $\Delta_\rcy \testf_r$ to denote jumps of $\testf_r$ at recycling times. Analogous notations are used for $\wk_r$ and $\rema$.

Miyazawa's RCL \citep{miyazawa_rate_1994} implies
\begin{equation}\label{eq:RCL-in-GG}
    \E{\dtime \testf_r} + \arate \E_\arv{\Delta_\arv \testf_r} + \rate_\rrcy\E_\rrcy{\Delta_\rrcy \testf_r} = 0.
\end{equation}
This RCL is simply describing the fact that the contribution of continuous changes ($\E{\dtime \testf_r}$) and jumps ($\E_\arv{\Delta_\arv \testf_r}$ and $\rate_\rrcy\E_\rrcy{\Delta_\rrcy \testf_r}$) in $\testf_r$ cancel out in the long-run average sense. 

To extract information about the G/G system, we analyze $\dtime \testf_r$,  $\Delta_\arv \testf_r$, and $\Delta_\rrcy \testf_r$ as below.
\begin{itemize}
    \item At all times, $\wk_r$ decreases at rate $- \dtime \wk_r = J_r$, and $\rema$ decreases at rate~$- \dtime \rema = 1$, so $\testf_r$ decreases at rate $- \dtime \testf_r = (J_r - \load_\rfresh)(\wk_r - \load_\rfresh \rema)$. 
    \item When a new arrival happens, $\testf_r$ jumps as follows. The new job contributes \r-work~$S_\rfresh$, so $\wk_r$ jumps up by $\Delta_\arv \wk_r = S_\rfresh$. And by definition, the new arrival happens just as $\rema$ reaches~$0$, at which point it jumps up to a newly sampled interarrival time~$A$, so $\Delta_\arv \rema = A$. This means that when a new arrival happens, $\testf_r$ jumps by $\Delta_\arv \testf_r = (S_\rfresh - \load_\rfresh A) \wk_r + \frac{1}{2} (S_\rfresh - \load_\rfresh A)^2$. 
    \item When an \r-recycling happens, the \r-work $\wk_r$ increases by $\Delta_\rrcy \wk_r = S_\rrcy$, and $\rema$ is unaffected, so $\Delta_\rrcy \testf_r = S_\rrcy (\wk_r - \load_\rfresh \rema)+ \frac{1}{2} S_\rrcy^2$. 
\end{itemize} 

Given the above formulas for $\dtime \testf_r$, $\Delta_\arv \testf_r$ and $\Delta_\rrcy \testf_r$, the terms in \eqref{eq:RCL-in-GG} can be computed one-by-one as follows. For $\E{\dtime \testf_r}$, we have
\[
    \E{\dtime \testf_r} 
    = -(1-\load_\rfresh)\E*{\wk_r} +\load_\rfresh(1-\load_\rfresh)\E*{\arri_\excess} + \E*{(1-J_r)\wk_r}  - \load_\rfresh \E*{(1-J_r)\rema}, 
\]
where we have used the fact that by basic renewal theory, $\E{\rema} = \E*{\arri_\excess}$. For $\arate \E_\arv{\Delta_\arv \testf_r}$, we have
\[
    \arate \E_\arv{\Delta_\arv \testf_r} 
    &= \arate \E_\arv*{(S_\rfresh - \load_\rfresh A) \wk_r + \frac{1}{2} (S_\rfresh - \load_\rfresh A)^2}\\
    &= \arate \E*{S_\rfresh - \load_\rfresh \arri } \E_\arv*{\wk_r}  + \frac{1}{2} \arate \E{(S_\rfresh -\load_\rfresh \arri)^2} \label{eq:wdl-arrival-term-key-calculation} \\
    &= \load_\rfresh \E*{(S_\rfresh)_\excess} + \load_\rfresh^2 \E*{\arri_\excess} - \load_\rfresh \E*{S_\rfresh},
\]
where the second equality is due to the fact that the new job's \r-work~$S_\rfresh$ and next interarrival time~$A$ are independent of the previous amount of \r-work~$\wk_r$, and the third equality uses \cref{def:excess} and the fact that $\E{S_r} = \load_\rfresh \E{A}$. Finally, for $\rate_\rrcy\E_\rrcy{\Delta_\rrcy \testf_r}$, we have
\[
    \rate_\rrcy\E_\rrcy{\Delta_\rrcy \testf_r} 
    = \rate_\rrcy \E_\rrcy*{\serv_\rrcy \wk_r} - \load_\rfresh \rate_\rrcy \E_\rrcy*{\serv_\rrcy \rema}+\load_\rrcy \E*{\gp{\serv_\rrcy}_\excess}.
 \]
Combining the three terms with \cref{eq:RCL-in-GG} completes the proof.
\end{proof}

\begin{remark}\label{remark:r-work_vs_Y-work}
    The proof of \cref{lem:work-decomp} does not depend on the details of the Gittins policy. It relies only partitioning the job states, with one part playing the role of states with rank less than~$r$. See \citet[Sections~7.2 and~8.3]{scully_new_2022} for an example of this with M/G arrivals.
\end{remark}

\begin{remark}\label{remark:rho-r-x-a-res}
The basic reason for multiplying $\rema$ by $\load_\rfresh$ in the definition of $\testf_r$ is that we want to avoid having a $\E_{\arv}{\wk_r}$ term in \eqref{eq:RCL-in-GG}, as it is likely to be a large and intractable term. As we can see from \eqref{eq:wdl-arrival-term-key-calculation}, when computing $\arate \E_\arv{\Delta_\arv \testf_r}$, the term involving $\E_{\arv}{\wk_r}$ vanishes because $\E{S_\rfresh - \load_\rfresh A} = \E{S_\rfresh}(1 - \lambda \E{A}) = 0$. Intuitively, ensuring $\wk_r - \load_\rfresh \rema$ has zero change in expectation when a new job arrives prevents the Palm expectation of arrivals from appearing in the RCL equation. This trick appears throughout the literature on applying the RCL to queues \citep{braverman_heavy_2017, miyazawa_diffusion_2015}.
\end{remark}

\section{Bounding Gittins's Suboptimality Gap}\label{sec:proofs}

In this section, we prove the main results using WINE and the work decomposition law introduced in Sections~\ref{sec:background} and \ref{sec:work-decomp}. We first derive a general formula that decomposes the suboptimality gap into four terms and analyze each term based on the specific settings of each theorem. We express the formula in terms of the following quantities.

\begin{definition}\label{def:cost-terms}
    We define \emph{residual interarrival cost}~$\mResidual$, \emph{recycling cost}~$\mRcy$, \emph{idleness cost}~$\mIdle$, and \emph{setup cost}~$\mSetup$ as follows:
    \[
        \mResidual &= \int_0^\infty \frac{-\load_\rfresh \E_\racc*{\rema}}{r^2}  \d{r},
        &
        \mIdle &= \int_0^\infty \frac{\E{(1-\bs_r-\bs_\setup)\wk_r} }{r^2(1-\load_\rfresh)} \d{r},
        \\
        \mRcy &=  \int_0^\infty \frac{\rate_\rrcy \E_\rrcy{\serv_\rrcy\wk_r}}{{r^2(1-\load_\rfresh)}} \d{r},
        &
        \mSetup &= \int_0^\infty \frac{\E{\bs_\setup \wk_r}}{r^2(1-\load_\rfresh)} \d{r}.
    \]
\end{definition}

\begin{restatable:lemma}[Decomposition of Performance Difference]\label{lem:perf-diff-decomp}
    The performance difference between the Gittins policy in \G/G/k/setup and any policy $\pi$ in the \G/G/1 (or \G/G/1/setup) can be decomposed as below:
    \[
        \E{N}_{\G/G/k/setup}^\gittins - \E{N}_{\G/G/1}^\pi = \gp[\big]{\mResidual^\gittins - \mResidual^\pi} + \mRcy^\gittins + \mIdle^\gittins + \gp[\big]{\mSetup^\gittins - \mSetup^\pi}. 
    \]
\end{restatable:lemma}

\begin{proof} 
    Only the last two terms in \cref{lem:work-decomp}, namely $\E_\racc*{\wk_r}$ and $\load_\rfresh \E_\racc*{\rema}$, depend on the specific scheduling policy. After expanding the definitions of the cost terms (\cref{def:cost-terms}) and $\E_\racc{\cdot}$ (\cref{sec:work-decomp:palm, eq:r-acc}), the result follows immediately from WINE (\cref{lem:wine}) and the fact that $\mRcy^\pi$ and $\mIdle^\pi$ are nonnegative.
\end{proof}

Note that we state \cref{lem:perf-diff-decomp} in terms of the Gittins policy only because our focus is the optimality of Gittins. The lemma is still true if we replace Gittins with any other policy.

Bounding Gittins's suboptimality gap thus reduces to bounding the four terms in \cref{lem:perf-diff-decomp}. We address one in each of \cref{sec:proofs:interarrival, sec:proofs:recycling, sec:proofs:idleness, sec:proofs:setup}, combining the bounds to prove our main results in \cref{sec:proofs:main}.

\subsection{Analysis of the Residual Interarrival Cost}
\label{sec:proofs:interarrival}

\begin{proposition}[Residual Interarrival Cost]\label{lem:remabd}
    For any policy~$\pi$,
    \[
        \mResidual^\gittins - \mResidual^\pi \leq \lambda(\remamax - \remamin). % \label{eq:residual-arrival-cost-bd}
    \]
\end{proposition}

\begin{proof}
    Observe that for deterministic~$v$, we have $\E_\racc{v} = v$ by the computation
    \[
        \E_\racc{v}
        = \frac{\E{(1-J_r) v} + \rate_\rrcy\E{\serv_\rrcy v}}{1-\load_\rfresh}
        = \frac{1 - \load_\rfresh - \load_\rrcy + \load_\rrcy}{1 - \load_\rfresh} v
        = v.
    \]
    The result follows from the fact that $\E_\racc{\rema} = \E_\racc{\E{\rema \given \agea}}$ and \cref{assump:remabd}.
\end{proof}

\subsection{Analysis of the Recycling Cost}
\label{sec:proofs:recycling}

\begin{proposition}[Recycling Cost]\label{lem:error-recycle-all}
    \leavevmode
    In the \G/G/k/setup, under the Gittins policy, we have
    \[
        \mRcy^\gittins \leq  (k-1) \log \frac{1}{1-\load}.
    \]
\end{proposition}

\begin{proof}
    The same bound has been shown for \M/G/k without setup times, e.g. \citep[Proposition~17.9]{scully_new_2022} and \citep[Lemma~B.5]{grosof_optimal_2022}. It turns out that the prior proofs rely only on the following fact:
    \begin{quote}
        Immediately before an \r-recycling, the number of \r-relevant jobs is at most $k - 1$.
    \end{quote}
    This fact still holds in the \G/G/k/setup under Gittins, so the same proof goes through. The fact holds because immediately before the recycling, the job that is about to be recycled is in service but is \r-irrelevant. If there were $k$ jobs that were \r-relevant, they would have priority under Gittins, preventing the \r-irrelevant job from being in service and thus preventing the \r-recycling.
\end{proof}

\subsection{Analysis of the Idleness Cost}
\label{sec:proofs:idleness}

\begin{restatable:proposition}[Idleness Cost]\label{lem:error-idle-all}
    In the \G/G/k/setup, under the Gittins policy, we have
    \[
        \mIdle^\gittins
        \leq (C - 1)(k - 1)\log\frac{1}{1 - \load} + (k - 1) \1(\P{U > 0} > 0),
    \]
    where $C = \frac{9}{8 \log 1.5} + 1 \approx 3.775$.
\end{restatable:proposition}

The proof of \cref{lem:error-idle-all} proceeds similarly to the proof of \citep[Proposition~17.6]{scully_new_2022}, but with a small modification to account for setup times. Given the similarity to prior work, we defer it to \cref{sec:proofs_extra}.

\subsection{Analysis of the Setup Cost}
\label{sec:proofs:setup}

\begin{proposition}[Single-Server Setup Cost]\label{lem:error-setup-single}
    In the \G/G/1/setup, the setup cost is fixed for any setup-non-idling policy. In particular, since $\gittins$ is also a setup-non-idling policy, we have
    \[
        \mSetup^\gittins - \mSetup^\pi = 0,
    \]
    for any other setup-non-idling policy $\pi$.
\end{proposition}

\begin{proof}[Proof of \Cref{lem:error-setup-single}]  
    Because $\bs_\setup \in \{0, 1\}$, we have
    \[
        \E{\bs_\setup \wk_r}^\pi = \E{\wk_r \given \bs_\setup = 1}^\pi \P{\bs_\setup=1}^\pi.
    \]
    Therefore, recalling \cref{def:cost-terms}, it suffices to show that both $\P{\bs_\setup=1}^\pi$ and $\E{\wk_r \given \bs_\setup = 1}^\pi$ do not depend on setup-non-idling policy~$\pi$. Observe that under any setup-non-idling policy, the distributions of busy periods (the continuous periods when there is work in the system) are unaffected by the order of serving specific jobs, and $J_{\setup}$ equals $1$ only during the first $S_{\setup}$ unit of time in each busy period, so $\P{\bs_\setup=1}^\pi$ do not depend on $\pi$. As for $\E{\wk_r \given \bs_\setup = 1}^\pi$, because the server cannot serve jobs when setting up and there is no \r-work in the system when the setup begins, $\wk_r$ is determined by the amount of \r-work that has arrived since the setup begins, whose distribution is independent of the policy.
\end{proof}

\begin{proposition}[Multiserver Setup Cost]\label{lem:error-setup}
    In the \G/G/k/setup, the setup cost under any setup-non-idling $\pi$ has the following bound:
    \[
        \mSetup^\pi \leq  \1(\P{U > 0} > 0) \gp[\big]{\arate k \E{\upti_\excess} + \arate \remamax + k - 1}. \label{eq:error-setup-bd-any-wc}
    \]
    In particular, the bound holds for Gittins.
\end{proposition}

To prove \cref{lem:error-setup}, we require a helper lemma. The lemma bounds the expected number of jobs in the system during a setup time. To state the lemma, we let $\bs_{\setup,i}$ be the indicator of whether server~$i$ is setting up and let $\upti_{\age,i}$ be the age of server~$i$'s setup process for each $i=1,2,\dots, n$. We set $\upti_{\age,i}$ to zero if server~$i$ is not setting up.

\begin{restatable:lemma}\label{lem:setup-model:number-during-setup}
    In the \G/G/k/setup, for any server~$i$ and all $a \geq 0$, we have
    \[
        \E{N \given J_{\setup,i}=1, \upti_{\age,i}=a} \leq \arate a + \arate \remamax + k - 1. \label{eq:number-during-setup}
    \]
\end{restatable:lemma}

The proof of \cref{lem:setup-model:number-during-setup} is nontrivial but uses standard techniques, so we defer it to \cref{sec:proofs_extra} and move on to proving \cref{lem:error-setup}.

\begin{proof}[Proof of \Cref{lem:error-setup}]
The case where $U = 0$ is clear, so we assume that $U > 0$ with nonzero probability. We first bound \r-setup error using the fact that $\load_\rfresh \leq \load$.
\[
    \label{eq:error-setup-relaxation}
    \mSetup
    = \int_0^\infty \frac{\E{\bs_\setup \wk_r}}{r^2(1-\load_\rfresh)} \d{r} \leq \int_0^\infty \frac{\E{\bs_\setup \wk_r}}{r^2(1-\load)} \d{r}
    = \int_0^\infty \frac{\E{\E{\bs_\setup \wk_r\given  X_1, \dots, X_N}}}{r^2(1-\load)} \d{r},
\]
where $X_1, X_2, \dots X_N$ are the states of all the jobs in the system (\cref{sec:model:scheduling}).
Using Tonelli's theorem and WINE (\cref{lem:wine}), the last expression can be rewritten as
\[
    \label{eq:error-setup-relaxation-wine}
    \int_0^\infty \frac{\E{\E{\bs_\setup \wk_r\given  X_1, \dots, X_N}}}{r^2(1-\load)} \d{r}
    =         \E*{\frac{\bs_\setup}{1-\load} \int_0^\infty \frac{\E{\wk_r\given X_1, \dots, X_N}}{r^2} \d{r}}     
    = \frac{\E{\bs_\setup N}}{1-\load}.
\]
By $J_\setup = \frac{1}{k}\sum_{i=1}^k J_{\setup,i}$, $\E{J_{\setup,i}N}=\E{\E{N\given J_{\setup,i}=1, \upti_{\age,i}}}$ and  \cref{lem:setup-model:number-during-setup},
\[
    \frac{\E{J_\setup N}}{1-\load} 
    &\leq \frac{\lambda}{k(1-\load)} \sum_{i=1}^k \E{J_{\setup,i} \upti_{\age,i}} + \frac{\lambda\remamax + k - 1}{k(1-\load)} \sum_{i=1}^k \E{J_{\setup,i}}.
    \label{eq:error-setup-progress}
\]
Now it remains to compute $\sum_{i=1}^k \E{J_{\setup,i}}$ and $\sum_{i=1}^k \E{J_{\setup,i} \upti_{\age,i}}$. The mean fraction of servers setting up is no more than the mean fraction of non-busy servers, which is $1-\load$, so
\[
    \label{eq:error-setup-avg-setup}
    \frac{1}{k} \sum_{i=1}^k \E{J_{\setup,i}} = \E{J_\setup} \leq 1-\load,
\]
Basic renewal theory and the $1/k$ service rate (\cref{sec:model:ggk}) imply the average age of a setup time is $k\E{\upti_\excess}$ (\cref{def:excess}), so
\[
    \label{eq:error-setup-avg-age}
    \frac{1}{k} \sum_{i=1}^k \E{J_{\setup,i}\upti_{\age,i}} = \sum_{i=1}^k \E{J_{\setup,i}} \, \E{\upti_\excess}.
\]
Combining \crefrange{eq:error-setup-relaxation}{eq:error-setup-avg-age} finishes the proof. 
\end{proof}

\subsection{Proofs of Main Results}
\label{sec:proofs:main}

\begin{proof}[Proof of \cref{thm:main_multiserver}]
    After expressing the suboptimality gap using \cref{lem:perf-diff-decomp}, we apply \cref{lem:remabd, lem:error-recycle-all, lem:error-idle-all, lem:error-setup} and use the fact that $\mSetup^\pi$ is non-negative. Grouping the $\log\frac{1}{1 - \load}$ terms to form~$\loss{(a)}$ and grouping the $\1(\P{U > 0} > 0)$ terms to form~$\loss{(c)}$ yields the result. 
\end{proof}

\begin{proof}[Proof of \cref{thm:main_single-server}]
    After expressing the suboptimality gap using \cref{lem:perf-diff-decomp}, we apply  \cref{lem:remabd, lem:error-recycle-all, lem:error-idle-all, lem:error-setup-single}. The only nonzero contribution comes from $\mResidual^\gittins - \mResidual^\pi$, which \cref{lem:remabd} shows to be at most $\loss{(b)}$.
\end{proof}

\section{Heavy-traffic Optimality}
\label{sec:heavy-traffic}

We now turn to prove \cref{thm:heavy_traffic}, which amounts to showing that Gittins's suboptimality gap, namely $\E{N}^\gittins - \inf_\pi \E{N}^\pi$ is small relative to the performance of the optimal policy, namely $\inf_\pi \E{N}^i$. It turns out that this is indeed the case: the suboptimality gap is small relative to the performance of SRPT in the \G/G/1, which is a lower bound on the performance of any policy in the \G/G/k/setup.

We first relate SRPT's performance in the \G/G/1 to its performance in the \M/G/1, which is known from prior work \citep{lin_heavy-traffic_2011}. We then use this result to prove \cref{thm:heavy_traffic}.

\begin{theorem}\label{lem:srpt-gg1-vs-mg1}
    Given \cref{assump:remabd}, in the heavy traffic limit, we have 
    \[
        \lim_{\load \to 1} \frac{\E{N}_{\G/G/1}^{\text{SRPT}}}{\E{N}_{\M/G/1}^{\text{SRPT}}}
        = \frac{c^2_S + c^2_A}{c^2_S + 1},
    \]
    where $c^2_V = \Var{V}/\E{V}^2$ and the two systems have the same service time distribution and average arrival rate. If $c^2_S = \infty$, then we interpret the right-hand side as~$1$.
\end{theorem}

\begin{proof}[Proof of \cref{lem:srpt-gg1-vs-mg1}]
    If $c^2_S = c^2_A = 0$, the result holds because $\E{N} \to \infty$ in the \M/G/1 but not the \G/G/1, which in this case are the M/D/1 and D/D/1, respectively (D is for ``deterministic''), implying the result. So we focus on the case where $c^2_S + c^2_A > 0$. We present the full proof for the $c^2_S < \infty$ case first, briefly sketching how to adapt the argument to the $c^2_S = \infty$ case at the end.
    
    Since SRPT is a special case of Gittins (\cref{ex:known_gittins}), we can analyze $\E{N}$ in the G/G/1 under SRPT using WINE and the work decomposition law. By \cref{lem:wine, lem:work-decomp, def:cost-terms}, we have
    \[ 
        \E{N} = \int_0^\infty  \frac{\load_\rfresh \gp*{\E{\gp*{S_\rfresh}_\excess} -  \E{S_\rfresh} + \E*{\arri_\excess}} + \load_\rrcy \E{\gp{\serv_\rrcy}_\excess}}{r^2(1-\load_\rfresh)} \d{r}
        + \mSetup+\mIdle+\mRcy+\mResidual.
    \]
    Since SRPT is a Gittins policy, by \cref{lem:remabd, lem:error-recycle-all, lem:error-idle-all, lem:error-setup-single}, we have that in the \G/G/1, $\mSetup=\mIdle=\mRcy=0$ and $\mResidual \in [-\lambda \remamax, -\lambda \remamin]$. We also note that under SRPT, the rank of a job is its remaining size, which together with \cref{def:excess} implies
    \[
        \E{S_\rfresh} &= \E{S \1(S \leq r)},
        &
        \load_\rfresh \E{\gp*{S_\rfresh}_\excess} &= \tfrac{\lambda}{2} \E{S^2 \1(S \leq r)},
        &
        \load_\rrcy \E{\gp{\serv_\rrcy}_\excess} &= \tfrac{\lambda}{2} r^2 \P{S > r}.
    \]
    Borrowing the notation of \citet{harchol-balter_performance_2013}, we write
    \[
        S_{\ol{r}} &= \min\{S, r\}
        &
        \rho_{\ol{r}} &= \lambda \E{S_{\ol{r}}},
    \]
    giving us
    \[
        \E{N} &= \int_0^\infty \frac{\load_{\ol{r}} \E{\gp*{S_{\ol{r}}}_\excess} + \load_\rfresh \gp{-\E{S_\rfresh} + \E*{\arri_\excess}}}{r^2(1-\load_\rfresh)} \d{r} + O(1). \label{eq:srpt-number-integral-var}
    \]
    in the \G/G/1 under SRPT. 
    Using \cref{def:excess} and the fact that $\E{S} = \load \E{A}$, we can compute
    \[
        \load (\E{S_\excess} - \E{S} + \E{A_\excess}) 
        = \tfrac{\lambda}{2}(\Var{S} + \Var{A}) + O(1 - \load).
    \]
    Reasoning similarly and using the fact that $\load_{\ol{r}} - \load_\rfresh \leq 1 - \load_{\rfresh}$, we have
    \[ 
        \load_{\ol{r}} \E{\gp*{S_{\ol{r}}}_\excess} + \load_\rfresh \gp{-\E{S_\rfresh} + \E*{\arri_\excess}}
        &= \tfrac{\lambda}{2}(\Var{S_{\ol{r}}} + \Var{A}) + O(1 - \load_r).
    \]
    Because $S_{\ol{r}} = \min\{S, r\}$, we have $\lim_{r\to\infty}\Var{S_{\ol{r}}} = \Var{S}$. Thus, for all $\epsilon > 0$, there exists $r^*$ such that for any $r\geq r^*$, we have $\abs{\Var{S_{\ol{r^*}}} - \Var{S}} \leq \epsilon$. Since $S$ does not depend on $\rho$, we can fix a sufficiently small constant $\epsilon$ so that $r^*$ is also a constant independent of $\rho$. Applying~\cref{eq:srpt-number-integral-var}, we can write $\E{N}$ as
    \[ 
        \E{N}
        = \int_0^{r^*} \frac{\load_{\ol{r}} \E{\gp*{S_{\ol{r}}}_\excess} + \load_\rfresh \gp{-\E{S_\rfresh} + \E*{\arri_\excess}}}{r^2(1-\load_\rfresh)} \d{r} + \int_{r^*}^\infty \gp[\bigg]{\frac{\frac{\lambda}{2}(\Var{S_{\ol{r}}} + \Var{A})}{1 - \load_r} + O(1)} \frac{1}{r^2} \d{r} + O(1).
    \]
    Observe that the first integral is non-negative, and it can be uniformly bounded at all loads by substituting $\rho_\rfresh \mapsto \E{S_\rfresh}/\E{S}$ and $\rho_{\ol{r}} \mapsto \E{S_{\ol{r}}}/\E{S}$, so it is $O(1)$.\footnote{%
        It is not a priori obvious that the integral converges due to the $r \to 0$ behavior. This can be verified by direct computation, but for our purposes, it suffices to use the prior knowledge that $\E{N}$ is finite under SRPT.}
    As for the second integral, note first that we can ignore the $O(1)$ since $\int_{r^*}^\infty \frac{O(1)}{r^2} \d{r} = O(1)$. Moreover, by our choice of $r^*$, we have
    \[
        \int_{r^*}^\infty \frac{\frac{\lambda}{2}(\Var{S_r} + \Var{A})}{r^2(1 - \load_r)} \frac{1}{r^2} \d{r} &= \tfrac{\lambda}{2}(\Var{S} + \Var{A} + \delta) \int_{r^*}^\infty \frac{1}{r^2(1 - \load_r)} \d{r},
    \]
    where $\delta \in [-\epsilon, \epsilon]$. Therefore, we have that in the \G/G/1 under SRPT, for some $\delta \in [-\epsilon, \epsilon]$,
    \[
        \E{N}_{\G/G/1}^{\text{SRPT}} = \tfrac{\lambda}{2}(\Var{S} + \Var{A} + \delta) \int_{r^*}^\infty \frac{1}{r^2(1 - \load_r)} \d{r} + O(1). \label{eq:srpt-n-final}
    \]
    Of course, the \M/G/1 is a special case of the \G/G/1, so \cref{eq:srpt-n-final} also holds for the \M/G/1, with $\Var{A} = \E{A}^2$ and $\delta$ replaced by some other $\delta' \in [-\epsilon, \epsilon]$. 
    \Citet[Theorem~5.8]{wierman_nearly_2005} show $\E{N}_{\M/G/1}^{\text{SRPT}} = \Omega\gp[\big]{\log \frac{1}{1 - \load}}$, so
    \[
        \frac{\E{N}^{\text{\G/G/1-SRPT}}}{\E{N}^{\text{\M/G/1-SRPT}}}
        = \frac{\Var{S} + \Var{A} + \delta}{\Var{S} + \E{A}^2 + \delta'} + O\gp[\bigg]{\frac{1}{\log \frac{1}{1 - \load}}}
        = \frac{\rho^2 c^2_S + c^2_A + \delta \lambda^2}{\rho^2 c^2_S + 1 + \delta' \lambda^2} + O\gp[\bigg]{\frac{1}{\log \frac{1}{1 - \load}}}.
    \]
    The result follows because $c^2_S$ and $c^2_A$ are independent of $\rho$, and $\delta, \delta' \in [-\epsilon, \epsilon]$ for arbitrarily small~$\epsilon$.
    
    We have proven the result assuming $c^2_S < \infty$. If instead $c^2_S = \infty$, then $\Var{S_r} \to \infty$ as $r \to \infty$. This means that for sufficiently large~$r$, the dominant term of $W_r$ is simply $\Var{S_r}/(1 - \load_r)$, which does not depend on~$A$ and is thus the same in the \G/G/1 and \M/G/1. One can use this fact to show that the performance ratio approaches~$1$ in heavy traffic by, as in the $c^2_S < \infty$ case, splitting the WINE integral at large~$r^*$, then observing that the $r > r^*$ part is dominant in heavy traffic.
\end{proof}

\begin{proof}[Proof of \Cref{thm:heavy_traffic}]
    It suffices to show that the suboptimality gap in \cref{thm:main_multiserver}, which is $O(1)$ for the $k = 1$ case and $O\gp[\big]{\frac{1}{1 - \load}}$ for the $k \geq 2$ case, is dominated by $\E{N}^\pi$ for any policy~$\pi$.
    
    We begin by observing that for any G/G arrival process, $\E{N}_{\G/G/1}^{\text{SRPT}}$ is a lower bound on $\E{N}^\pi$ for any policy~$\pi$. This is because we can view the \G/G/k/setup as a version of a \G/G/1 that imposes extra constraints on the scheduler (\cref{sec:model:scheduling}), and SRPT minimizes $\E{N}$ in the \G/G/1 \citep{schrage_proof_1968}.

    Next, we observe that by \cref{lem:srpt-gg1-vs-mg1} and our assumption that $c^2_S + c^2_A > 0$, SRPT's \G/G/1 heavy-traffic performance is within a constant factor of its \M/G/1 heavy-traffic performance.
    
    It thus suffices to show a lower bound on $\E{N}_{\M/G/1}^{\text{SRPT}}$. When $k = 1$, we need only an $\omega(1)$ bound, which always holds \citep[Theorem~5.8]{wierman_nearly_2005}. When $k \geq 2$, we need an $\omega\gp[\big]{\log\frac{1}{1 - \load}}$ bound, which prior work \citep[Proof of Theorem~1.3 in Appendix~B.2]{scully_gittins_2020} shows to hold if $\E{S^2 (\log S)^+} < \infty$.
\end{proof}

%% file: extensions-and-conclusion.tex
\section{Potential Extensions}
\label{sec:extensions}

We have seen that combining our new work decomposition law (\cref{lem:work-decomp}) with WINE (\cref{lem:wine}) enables the analysis of systems with many complex features, such as the \G/G/k/setup. Thanks to the generality of both systems, we could apply the same technique even beyond the \G/G/k/setup. This section sketches how this can be done for three features: multiserver jobs, batch arrivals, and generalized vacations. We emphasize that our goal here is not to give full proofs, but rather to demonstrate the applicability of our technique to additional systems. We thus say ``should be'' rather than ``is'' when stating the end results.

\subsection{Multiserver Jobs}
\label{sec:extensions:multiserver_jobs}

We study a variation of the model of \citet{grosof_optimal_2022}. We consider a variant of our \G/G/k model where each job has a \emph{server need}~$m(x)$, which is a function of its state~$x$. Whenever a job in state~$x$ runs, it must occupy exactly $m(x)$ servers. It is thus served at rate~$m(x)/k$, thanks to our convention of servers operating at rate~$1/k$ (\cref{sec:model:ggk}). We refer to this model as the \emph{\G/G/k/MSJ}, where MSJ stands for ``multiserver job''. \Citet{grosof_optimal_2022} study what we would call the \M/G/k/MSJ.\footnote{%
    \Citet{grosof_optimal_2022} actually consider a slightly more restrictive case in which jobs' server needs remain constant throughout service, though their proofs could be straightforwardly generalized to handle dynamically changing server needs. The main novelty of our discussion is thus the extension to G/G arrivals.}

One of the main challenges when scheduling in MSJ systems is that it is no longer clear how to stabilize the system. Indeed, analyzing stability even in \M/M/k/MSJ systems is an area of current research \citep{grosof_stability_2020, rumyantsev_stability_2017}, and optimal scheduling in these systems is an open problem. However, \citet{grosof_optimal_2022} show that if every possible server need~$m(x)$ is a divisor of the number of servers~$k$, then one can ensure stability with a procedure called \emph{DivisorFilling}.

The DivisorFilling procedure takes as input any set of~$k$ jobs, then outputs a subset of those jobs whose server needs sum to exactly~$k$. DivisorFilling can be combined with the Gittins policy by passing the $k$ jobs of least Gittins rank to DivisorFilling, resulting in a policy called \emph{DivisorFiling-Gittins} \citep{grosof_optimal_2022}.

The \G/G/k/MSJ under DivisorFilling-Gittins can be analyzed in much the same way as the \G/G/k under Gittins. Recalling the structure of the latter analysis from \cref{sec:proofs}, we encounter the same four ``cost terms'' (\cref{def:cost-terms}) to bound.
\* The residual interarrival cost~$\mSetup$ can be bounded exactly as in \cref{lem:remabd}.
\* The recycling cost~$\mRcy$ can be bounded exactly as in \cref{lem:error-recycle-all}, because DivisorFilling-Gittins ensures the key property that makes the proof work \citep[Lemma~B.5]{grosof_optimal_2022}.
\* The idleness cost~$\mIdle$ can be bounded by following the same steps as \citep[Lemma~B.3]{grosof_optimal_2022}, because their proof does not rely on Poisson arrivals. The resulting bound is $\mIdle \leq e (k - 1) \ceil[\big]{\log \frac{1}{1 - \load}}$.\footnote{%
    In the $\load \to 1$ limit, this bound below is actually slightly better than that in \cref{lem:error-idle-all}. See \cref{sec:results:remarks} for discussion.}
\* The setup cost~$\mSetup$ is zero, as there are no setup times.
\*/
We thus see that the suboptimality gap of DivisorFilling-Gittins in the \G/G/k/MSJ should be at most
\[
    \E{N}_{\G/G/k\text{/MSJ}}^{\text{DivisorFilling-Gittins}} - \inf_\pi \E{N}_{\G/G/1}^\pi
    \leq e (k - 1) \ceil*{\log \frac{1}{1 - \load}} + \lambda(\remamax - \remamin).
\]
This is analogous to what \cref{thm:main_multiserver} says about the \G/G/k, as $\loss{(a)}$ is replaced by $e (k - 1) \ceil[\big]{\log \frac{1}{1 - \load}}$.

It is likely that, for a suitable definition of setup times in an MSJ model, one could analyze the same system with setup times, obtaining a result analogous to the one we obtain for the \G/G/k/setup.

\subsection{Batch Arrivals}
\label{sec:extensions:batch_arrivals}

\Citet{scully_gittins_2021} introduce a general model of batch arrivals that we call \emph{batch-M/G arrivals}. The main notable feature of the model is that it makes few assumptions about what batches look like. For example, it may be that the initial states of jobs in the same batch are correlated with each other. We consider the same type of batches but allow general batch interarrival times, resulting in \emph{batch-G/G} arrivals. All of the results in \cref{sec:results} should generalize to batch-G/G arrivals. There are three changes needed for the proof, and only the third impacts the end results.

First, one needs to modify the definitions of load, \r-fresh load, and other concepts related to the arrival process. But these are straightforward changes. The most important note here is that \cref{assump:remabd} should refer to the batch interarrival time.

Second, batch arrivals affect the work decomposition law (\cref{lem:work-decomp}). But they only affect the term that is common to all systems, namely the numerator of the $1/(1 - \load_r)$ term. By \cref{lem:perf-diff-decomp}, this change does not affect suboptimality gaps, which is the basis of all of the results in \cref{sec:results}.

Third, batch arrivals affect the setup cost~$\mSetup$ in multiserver systems. Specifically, we need to slightly modify the statement and proof of \cref{lem:error-setup} to account for the fact that multiple jobs can arrive at once. If more jobs arrive than there are idle servers, this means a single setup is effectively triggered by multiple jobs. The end result is that one must incorporate a term related to the batch size distribution. In contrast, the setup cost in single-server systems is unaffected.

Taken together, these observations imply that our \G/G/k and \G/G/1/setup results should generalize immediately to batch arrivals. With some effort, the \G/G/k/setup results should also generalize.

\subsection{Generalized Vacations}
\label{sec:extensions:generalized_vacations}

The term \emph{generalized vacations} refers to a range of models where servers may be unavailable, including:
\* Setup times, as studied in this work. This includes models beyond ours, e.g. where we make different decisions about when to start setting up a server, or where a setup time can be canceled if the job that triggered it enters service at another server.
\* \emph{Vacations}, where whenever the server goes idle, it goes on vacation for a given amount of time, then only begins serving jobs again when it returns.
\* \emph{Server breakdowns}, where servers can become unavailable in the middle of serving a job.
\* \emph{Threshold policies}, where servers stay idle until there are a given number of jobs in the system.
\*/
These are only a few examples of what generalized vacations can model \citep{fuhrmann_stochastic_1985, doshi_queueing_1986, miyazawa_decomposition_1994}.

One can, in principle, bound Gittins's suboptimality gap in the \G/G/k with generalized vacations using essentially the same approach we take for the \G/G/k/setup. The main change is that we now interpret $J_\setup$ as the fraction of servers that are unavailable, so we now think of $\mSetup$ as an \emph{unavailability cost}.

Of course, whether bounding $\mSetup$ is tractable to bound depends on the specifics of the model. As in the proofs of \cref{lem:error-setup, lem:error-setup-single}, the key question is: how many jobs might there be in the system while a server is unavailable? Sometimes, this will be very hard to bound, e.g. for server breakdowns. But in other cases, the bound is nearly immediate. For example, consider a threshold policy that does not start serving jobs until there are~$n$ jobs present, at which point it serves jobs until the system empties. We would then have $\mSetup \in [0, n]$ under any scheduling policy.

One important application of generalized vacations is to more general setup time models. For instance, in practice, it is helpful to not turn servers off as soon as they become idle. One can imagine a wide range of \emph{power management} policies controlling when servers turn on and off. Provided we do not wait too long to set up servers while there are jobs in the queue, $\mSetup$ should not be too large, in which case Gittins would have a small suboptimality gap. This means that, in some sense, the power management and job scheduling problems are orthogonal, because a single scheduling policy, namely Gittins, performs well for a wide range of power management policies.

\section{Conclusion}
\label{sec:conclusion}

This work presents the first analysis of the Gittins policy in the \G/G/k/setup. We prove simple and explicit bounds on Gittins's suboptimality gap, which are tight enough to imply that Gittins is optimal in heavy traffic in the \G/G/k/setup. As a corollary, we find that Gittins is optimal in the \M/G/1/setup. Prior to these results, Gittins had not been analyzed in even the \G/G/1, let alone the \G/G/k/setup.

There are several ways in which one might hope to improve our bounds. This is especially true in light traffic, namely the $\load \to 0$ limit. Here we have a constant suboptimality gap for mean number of jobs, but by Little's law \citep{little_littles_2011}, this corresponds to an \emph{infinite} suboptimality gap for mean response time. We conjecture that Gittins's mean response time suboptimality gap remains bounded in light traffic, but there are significant obstacles to doing so, related to the notorious problem of analyzing the idle period of the \G/G/1 \citep{wolff_idle_2003, li_characterizing_1995}.

Our theoretical results also raise several questions that could be studied with simulations. One such question is related to the additive structure of our suboptimality gap bound in \cref{thm:main_multiserver}, in which each of (a)~multiple servers, (b)~non-Poisson arrivals, and (c)~setup times contributing to the bound via a separate term. If we simulate Gittins in systems with various mixtures of (a), (b), and~(c), do we observe an analogous (approximate) additive structure in its empirical performance? We hypothesize the answer is ``yes'', because each of the terms in \cref{thm:main_multiserver} has a distinct cause, and we suspect the interactions between these causes are relatively weak. Investigating this is an interesting direction for future work.

Taking a step back, we might ask: should one use Gittins to minimize mean number of jobs in practice, even beyond the \G/G/k/setup modeling assumptions? While this is clearly a question larger than we can definitively answer, we believe that our main results, the potential extensions sketched in \cref{sec:extensions}, and other recent work on Gittins and SRPT in multiserver systems \citep{grosof_srpt_2018, grosof_load_2019, scully_gittins_2020, scully_new_2022, grosof_optimal_2022} point towards ``yes''. Even though the currently known theoretical bounds on Gittins and SRPT are not tight, we have no comparable bounds for other policies, aside from a few close relatives of SRPT \citep{grosof_srpt_2018}. The mere existence of these bounds is thus a point in favor of Gittins. But we are still in the early years of understanding multiserver scheduling.

%% file: appendix.tex
\section{Summary of Notation}
\label{sec:notation}

\paragraph*{Core \G/G/k/setup model (\cref{sec:model:ggk, sec:model:setup_times})}
\*[beginpenalty=10000] $k$: number of servers
\** Each server has speed $1/k$, so the total service capacity is always $k \cdot 1/k = 1$
\* $S$: job size random variable
\** With $k$ servers, job of size~$S$ results in $k S$ service time
\* $U$: setup work random variable (\cref{sec:model:setup_times})
\** With $k$ servers, setup work~$U$ results in $k U$ setup time
\* $A$: interarrival time random variable
\* $\lambda = 1/\E{A}$: arrival rate
\* $\load = \lambda \E{S} < 1$: load
\* $N$: number of jobs in the system
\* $\E{N}_{\text{SYS}}^\pi$: mean number of jobs in system~SYS under policy~$\pi$
\** To reduce clutter, we omit ``SYS'' and/or ``$\pi$'' when there is no ambiguity
\*/
\paragraph*{Remaining interarrival time (\cref{assump:remabd})}
\*[beginpenalty=10000] $\rema$: residual interarrival time, i.e. time until next arrival
\* $\agea$: interarrival age, i.e. time since last arrival
\* $\remamin, \remamax$: bounds on $\E{\rema \given \agea}$
\*/
\paragraph*{Markov-process job model (\cref{sec:model:scheduling})}
\*[beginpenalty=10000] See \cref{ex:known_model, ex:unknown_model} for concrete examples of the Markov-process job model
\* $\mathbb{X}$: space of possible job states
\* $X(t)$: state of a job that has received $t$ service so far
\* $X_i$: state of the $i$th job currently in the system ($i \in \{1, \dots, N\}$)
\* $\top$: terminal state, i.e. a job completes once $X(t) = \top$
\* $S(x)$: remaining work of a job in state~$x$
\* $W = \sum_{i = 1}^N S(X_i)$: (total) work in the system
\*/
\paragraph*{Gittins policy (\cref{sec:model:gittins})}
\*[beginpenalty=10000] $\gittins$: abbreviation for Gittins in formulas, e.g. $\E{N}^\gittins$
\* $\rank_\gittins(x)$: rank, i.e. priority (lower is better), of job in state~$x$
\** See \cref{ex:known_gittins, ex:unknown_gittins} for concrete examples of the rank function
\*/
\paragraph*{Main results (\cref{sec:results})}
\*[beginpenalty=10000] $\loss{(a)}, \loss{(c)}, \loss{(c)}$: suboptimality due to (a)~multiple servers, (b)~non-Poisson arrivals, (c)~setup times
\* $c^2_V = \frac{\Var{V}}{\E{V}^2}$: squared coefficient of variation of~$V$
\*/
\paragraph*{WINE and \r-work (\cref{sec:background})}
\*[beginpenalty=10000] $S_r(x)$: remaining \r-work of a job in state~$x$
\* $W_r = \sum_{i = 1}^N S_r(X_i)$: \r-work in the system
\* \r-fresh, \r-recycled: see \cref{def:r-fresh, def:r-recycled}
\* $S_\rfresh = S_r(X(0))$: remaining \r-work of a newly arrived job
\* $\rho_\rfresh = \lambda \E{S_\rfresh}$: \r-fresh load
\* $\lambda_\rrcy$: average rate of \r-recyclings
\* $S_\rrcy$: remaining \r-work of a newly \r-recycled job
\* $\rho_\rrcy = \lambda_\rrcy \E{S_\rrcy}$: \r-recycled load
\* $J_r$: fraction of servers that are serving jobs of rank at most~$r$
\* $J_\setup$: fraction of servers that are setting up
\*/
\paragraph*{Stating work decomposition law (\cref{sec:work-decomp:palm})}
\*[beginpenalty=10000] $\E_\arv{\cdot}$: expectation sampled immediately before an arrival
\* $\E_\rrcy{\cdot}$: expectation sampled immediately before an \r-recycling
\* $\E_\racc{\cdot}$: an expectation defined in terms of $\E_\rrcy{\cdot}$ in \cref{eq:r-acc}
\* $V_\excess$: excess of random variable $V$
\** Key property: $\E{V_\excess} = \frac{\E{V^2}}{2 \E{V}}$
\*/
\paragraph*{Proofs of main results (\cref{sec:proofs})}
\*[beginpenalty=10000] $\mResidual^\pi$: contribution to $\E{N}^\pi$ due to residual interarrival time, bounded in \cref{lem:remabd}
\* $\mRcy^\pi$: contribution to $\E{N}^\pi$ due to recyclings, bounded in \cref{lem:error-recycle-all}
\* $\mIdle^\pi$: contribution to $\E{N}^\pi$ due to idle servers, bounded in \cref{lem:error-idle-all}
\* $\mSetup^\pi$: contribution to $\E{N}^\pi$ due to setup times, bounded in \cref{lem:error-setup, lem:error-setup-single}
\*/

\section{Additional Model Details}
\label{sec:model_extra}

\subsection{Processor Sharing}
\label{sec:model_extra:scheduler_actions}

Thus far, we have discussed the scheduler as assigning jobs to servers in one-to-one fashion. However, because we are in a preemptive setting, it is possible for the scheduler to effectively share a server between multiple jobs by rapidly switching between them. As such, we explicitly allow the scheduler to use \emph{processor sharing}, simultaneously serving multiple jobs at reduced rates. Processor sharing arises naturally when using Gittins \citep{aalto_gittins_2009, aalto_properties_2011}.

Exactly how processor sharing works with multiple servers and setup times has some subtle corner cases, so for completeness, we give the details below. But for ease of discussion, we do not concern ourselves with these details during the proofs, because they complicate the arguments without adding insight. With that said, modifying the proofs to account for processor sharing is a straightforward exercise. For instance, instead of referring to the job at server~$i$, we might refer to an appropriately weighted random choice among jobs currently sharing server~$i$.

In a system without setup times, the scheduler's action is a \emph{service rate vector} $\alpha = (\alpha_1, \dots, \alpha_N)$, where recall that $N$ is the number of jobs. Here $\alpha_i$ is the rate at which job~$i$ is being served. In the \G/G/k, the fact that we have $k$ servers of service rate~$1/k$ corresponds to $\alpha$ obeying the constraints $\sum_{i = 1}^N \alpha_i \leq 1$ and $\alpha_i \in \sqgp{0, 1/k}$. From this perspective, a non-idling policy is one such that no $\alpha_i$ can be increased without violating a constraint. To clarify, the service rate vector~$\alpha$ is a \emph{gradual} control: it is chosen at every moment in time and can be changed at will.

In a system with setup times, the situation is similar, but we can only assign jobs to busy servers. The first constraint on $\alpha$ thus changes to $\sum_{i = 1}^N \alpha_i \leq K_{\mathrm{busy}}/k$, where $K_{\mathrm{busy}}$ is the number of busy servers.

\subsection{General Definition of the Gittins Rank Function}
\label{sec:model_extra:gittins_rank}

How does Gittins decide what rank to assign each job? The main idea is to assign jobs low rank if they are likely to finish with a small amount of service. To formalize this intuition, consider a job in state~$x$, and suppose we were to serve the job until it enters some state in set $\mathbb{Y} \subseteq \mathbb{X}$. This would involve serving the job for some amount of time, and there would be some chance the job completes. Let
\[
    S_{\mathbb{Y}}(x)
    &= \text{amount of service a job starting at~$x$ needs to finish or enter $\mathbb{Y}$} \\*
    \label{eq:service_until}
    &= \gp[\big]{\inf\curlgp{s \geq 0 \given X(s) \in \mathbb{Y} \cup \{\top\}} \given X(0) = x}, \\
    C_{\mathbb{Y}}(x)
    &= \text{event that a job starting at~$x$ finishes before entering $\mathbb{Y}$} \\*
    \label{eq:completion_event}
    &= \sqgp[\big]{X\gp[\big]{\inf\curlgp{s \geq 0 \given X(s) \in \mathbb{Y} \cup \{\top\}}} = \top \given X(0) = x}.
\]
We can think of $\E{S_{\mathbb{Y}}(x)}/\P{C_{\mathbb{Y}}(x)}$ as a time-per-completion ratio for jobs in state~$x$, as measured by serving them until they finish or enter~$\mathbb{Y}$. The key idea of the Gittins policy is to let a job's rank be its best possible time-per-completion ratio:\footnote{%
    We consider the ratio to be $\infty$, and thus not the infimum, if $\P{C_{\mathbb{Y}}(x)} = 0$.}
\[
    \label{eq:rank_gittins}
    \rank_\gittins(x) = \inf_{\mathbb{Y} \subseteq \mathbb{X}} \frac{\E{S_{\mathbb{Y}}(x)}}{\P{C_{\mathbb{Y}}(x)}}.
\]

As an example, consider the case of known job sizes, where a job's state~$x$ is its remaining work. The best completion ratio is achieved by always running the job until it finishes, i.e. $\mathbb{Y} = \emptyset$. This results in $S_\emptyset(x) = x$ and the event $C_\emptyset(x)$ occurring with probability~$1$, so $\rank_\gittins(x) = x$. That is, Gittins always serves the jobs of least remaining work, so it reduces to SRPT (\cref{ex:known_gittins}).

\subsection{Technical Considerations for the Job Markov Process}
\label{sec:model_extra:technicalities}

We have omitted some technical details and assumptions in our definition of the Gittins policy. For example, in 
\crefrange{eq:service_until}{eq:rank_gittins},
we must restrict attention to sets $\mathbb{Y} \subseteq \mathbb{X}$ such that $S_{\mathbb{Y}}(x)$ and $C_{\mathbb{Y}}(x)$ are measurable with respect to the natural filtration on the job Markov process. More generally, the theory of the Gittins policy relies on being able to solve an optimal stopping problem which is known in the Markov-process job model literature as the ``Gittins game'' \citep{scully_gittins_2020, scully_gittins_2021, scully_new_2022}.

Following the convention of prior literature, we consider the technical foundations of the Gittins game to be outside the scope of this paper. Our results apply to any job Markov process where the foundations can be established. As explained by \citet[Appendix~D]{scully_gittins_2020}, this is not a restrictive assumption.

\section{Deferred Proofs}
\label{sec:proofs_extra}

\restate*{lem:error-idle-all}

\begin{proof}
    At a high level, we follow the same main as the proof of \citep[Proposition~17.6]{scully_new_2022}, but one of the steps requires modification to account for setup times.
    
    By following the same argument as \citep[Lemma~17.4]{scully_new_2022}, one can show that
    \[
        \mIdle^\gittins = \int_0^\infty \frac{\E{(1-\bs_r-\bs_\setup)\wk_r}}{r^2(1-\load_\rfresh)} \d{r} \leq (k-1)\inf_{n\in \Z_{\geq 1}} n \gp*{\frac{1}{1-\load}}^{1/n}.
        \label{eq:setup-model:idleness-term-bound-17-10}
    \]
    The key fact is that when $1-\bs_r - \bs_\setup > 0$, there is an idle server, so there are at most $k-1$ jobs in the system, as otherwise a server would start setting up. Given this fact, the rest of the proof carries through verbatim. We can also show an alternative bound for $\mIdle^\gittins$ that is tighter in light traffic:
    \[
         \mIdle^\gittins = \int_0^\infty \frac{\E{(1-\bs_r-\bs_\setup)\wk_r}}{r^2(1-\load_\rfresh)} \d{r} \leq (k - 1)\frac{\load}{1 - \load} + (k - 1)\1(\P{U > 0} > 0). \label{eq:setup-model:idleness-term-bound-17-4}
    \]
    We prove \cref{eq:setup-model:idleness-term-bound-17-4} below. Combining \cref{eq:setup-model:idleness-term-bound-17-10} and \cref{eq:setup-model:idleness-term-bound-17-4} yields
    \[
        \mIdle^\gittins
        &\leq (k - 1) \min\curlgp[\bigg]{\inf_{n\in \Z_{\geq 1}} n \gp[\bigg]{\frac{1}{1-\load}}^{1/n}, \frac{\load}{1 - \load} + \1(\P{U > 0} > 0)} \\
        &\leq (k - 1) \min\curlgp[\bigg]{\inf_{n\in \Z_{\geq 1}} n \gp[\bigg]{\frac{1}{1-\load}}^{1/n}, \frac{\load}{1 - \load}} + (k - 1) \1(\P{U > 0} > 0) \\
        &\leq (C - 1)(k - 1)\log\frac{1}{1 - \load} + (k - 1) \1(\P{U > 0} > 0),
    \]
    where the last inequality follows from the same computation as in \citep[Proposition~17.6]{scully_new_2022}.
    
    It remains only to prove \cref{eq:setup-model:idleness-term-bound-17-4}. The proof follows the same steps as \citep[Lemma~17.5]{scully_new_2022}, but setup times add some additional considerations. For the purposes of this proof, we abuse notation slightly by defining $X_i = \top$ for the state of job~$i$ if there is no $i$th job. We then define the remaining \r-work of~$\top$ to be $S_r(\top) = 0$.
    
    We begin by noting that if $\bs_r + \bs_\setup < 1$, then there are less than~$k$ jobs in the system, so
    \[
        \label{eq:setup-model:idleness-term-bound-17-6}
        \E{(1-\bs_r-\bs_\setup)\wk_r} &= \sum_{i=1}^k \E{S_r(X_i) (1-\bs_r-\bs_\setup)}
        \leq  \frac{k - 1}{k} \sum_{i=1}^k \E{S_r(X_i)},
    \]
    where the last equality follows from the fact that $\bs_r + \bs_\setup < 1$ is an integer multiple of~$1/k$. Applying \cref{lem:single-job-wine, eq:setup-model:idleness-term-bound-17-6} to the definition of $\mIdle^\gittins$ yields
    \[
        \mIdle^\gittins
        &= \int_0^\infty \frac{\E{(1-\bs_r-\bs_\setup)\wk_r}}{r^2(1-\load_\rfresh)} \d{r} \\
        &\leq \frac{k-1}{k(1-\load)}\int_0^\infty  \sum_{i=1}^k \frac{\E{S_r(X_i)}}{r^2}\d{r} \\
        &\leq \frac{k-1}{k(1-\load)} \sum_{i=1}^k \P{X_i \neq \top} \\
        &= \frac{(k - 1) \E{\min\curlgp{N, k}}}{k(1 - \load)}. \label{eq:setup-model:idleness-term-bound-17-7}
    \]
    The fact that we are considering the setup-non-idling version of Gittins (\cref{sec:model:setup_times}) implies that the number of non-idle servers (i.e. busy or setting-up) is always at least $\min\curlgp{N, k}$. So it suffices to bound the mean fraction of non-idle servers.
    \* The mean fraction of busy servers is~$\load$.
    \* The mean fraction of setting-up servers is at most $(1 - \load) \1(\P{U > 0} > 0)$.
    \** If there are no setup times, it is clearly zero.
    \** If there are setup times, it is at most the fraction of non-busy servers, namely $1 - \load$.
    \*/
    This means
    \[
        \frac{(k - 1) \E{\min\curlgp{N, k}}}{k(1 - \load)}
        \leq (k - 1)\frac{\load}{1 - \load} + (k - 1) \1(\P{U > 0} > 0).
        \label{eq:setup-model:idleness-term-bound-17-8}
    \]
    Combining \cref{eq:setup-model:idleness-term-bound-17-7} and \cref{eq:setup-model:idleness-term-bound-17-8} yields \cref{eq:setup-model:idleness-term-bound-17-4}, as desired.
\end{proof}

\restate*{lem:setup-model:number-during-setup}

\begin{proof}
    Observe that during setup, the number of jobs in the system can be divided into three sets: jobs present when setup started, jobs triggered the setup and jobs that have arrived since setup started. 
    Therefore, to bound $\E{N \given J_{\setup,i}=1, \upti_{\age,i}=a}$, it is sufficient to bound the total number of jobs in the three sets on the right. Just before the setup starts, there is at least one idle server, so there are at most $k-1$ jobs present. And there is exactly $1$ job that triggers the setup. Let $N_{\arv}(a)$ be the number of jobs that have arrived since $a$ unit of time ago, then the above argument shows that
    \[
        \E{N\given J_{\setup,i}=1, \upti_{\age,i}=a} &\leq \E{N_{\arv}(a) \given J_{\setup,i}=1, \upti_{\age,i}=a} + k. \label{eq:number-during-setup-vs-arrivals}
    \]
    Therefore, to prove the lemma, it is sufficient to show that for any $a\geq 0$,
    \[
        \E{N_{\arv}(a) \given J_{\setup,i}=1, \upti_{\age,i}=a} \leq \arate a + \arate \remamax - 1. \label{eq:arrivals-during-setup-goal}
    \] 
    
    Conditioned on $J_{\setup,i}=1$ and $\upti_{\age,i}=a$, we know that server $i$ started setup process $a$ unit of time ago. Since the setup process is triggered by an arrival, if we denote the sequence of interarrival times since the setup as $A_1, A_2, \dots, A_n,\dots$, then we have the equation
    \[
        \E*{\sum_{n=1}^{N_{\arv}(a)+1} A_{n} \given J_{\setup,i}=1, \upti_{\age,i}=a} = \E{a + \rema \given J_{\setup,i}=1, \upti_{\age,i}=a}. \label{eq:total-arrival-time-equals-age-plus-res}
    \]
    The right-hand side of \eqref{eq:total-arrival-time-equals-age-plus-res} can be upper bounded using Assumption~\ref{assump:remabd} and the fact that $\rema$ is conditionally independent of the past given the age:
    \[
        \E{a + \rema \given J_{\setup,i}=1, \upti_{\age,i}=a} = \E{a+\rema \given \upti_{\age,i}=a} \leq a+\remamax. \label{eq:total-arrival-time-equation-rhs}
    \]
    To bound the left-hand side of \eqref{eq:total-arrival-time-equals-age-plus-res}, observe that $J_{\setup,i}=1$ and $\upti_{\age,i}=a$ is equivalent to having a setup process start $a$ unit of time ago and last until now, which is determined by the setup time and is independent of arrival process in the recent $a$ unit of time, so we have
    \[
         \E*{\sum_{n=1}^{N_{\arv}(a)+1} A_{n} \given J_{\setup,i}=1, \upti_{\age,i}=a} =  \E*{\sum_{n=1}^{N_{\arv}(a)+1} A_{n}}.\label{eq:total-arrival-time-equation-lhs-independence}
    \]
    We further observe that $A_1, A_2, \dots, A_n, \dots$ are a sequence of i.i.d. variables, and $N_{\arv}+1$ is the smallest $n$ such that the partial sum $A_1 + \dots A_n \geq a$, by Wald's equation,
    \[
        \E*{\sum_{n=1}^{N_{\arv}(a)+1} A_{n}} = \E{N_{\arv}(a)+1}\E{A} = \frac{\E{N_{\arv}(a)+1}}{\lambda}. \label{eq:total-arrival-time-equation-lhs-ward-equation}
    \]
    Combining \crefrange{eq:total-arrival-time-equals-age-plus-res}{eq:total-arrival-time-equation-lhs-ward-equation}, we get $\E{N_{\arv}(a)+1}/\lambda \leq a + \remamax$, implying \eqref{eq:arrivals-during-setup-goal}.
\end{proof}

\section{Towards Proving Stability of the \G/G/k/setup under Work-Conserving Scheduling Policies}
\label{sec:stability_sketch}

In this section, we sketch some initial ideas that we believe could be used to prove the stability of the \G/G/k/setup under Gittins and other complex scheduling policies. For simplicity, we focus on the \G/G/k under stationary scheduling policies.

\subsection{Partial Proof Sketch}

Assume we are using some scheduling policy that makes the system a Markov process. Consider the embedded Markov process of our system at arrival instants. We construct a Lyapunov function and show its negative drift for this embedded Markov process when the total work in the system or the largest job size is outside a bounded region. The negative drift of the Lyapunov function should imply the stability of the system after applying a certain version of the Foster-Lyapunov theorem, although there are still some technical obstacles, as we discuss in the next section. 

At any time~$t\geq 0$, we let $W(t)$ be the total work in the system, let $\Sbig(t)$ be the size of the largest job in the system, and let $\rema(t)$ be the residual arrival time.

Fix large constants $\alpha > 0$ and $\beta > 0$.
Define the Lyapunov function $V = W - \rho \rema + k (\Sbig - \alpha)^+$. Note that at each arrival point, $\rema=0$ and $V\geq 0$. 
The continuous rate of change of $V$ is 
\[
    \dot{V} &= -(1-\rho) + 1-J - k \indic{\Sbig>\alpha} \dot{S}_{\text{big}}(t) \\
    &\leq -(1-\rho) + 1-J - \indic{\Sbig>\alpha, J<1} \\
    &\leq -(1-\rho) + (1-J)\indic{\Sbig \leq \alpha},
\]
where in the first inequality, we use the fact that when $J<1$, all jobs are in service, so the size of the largest job decreases at a rate $1/k$; and in the second inequality, we use the fact that $\indic{J<1}\geq 1-J$.
The expected jump of $V$ at an arrival point is
\[
    \Delta V &= \E{S - \rho A + k\max((\Sbig-\alpha)^+, (S-\alpha)^+) - k(\Sbig-\alpha)^+} \\
    &\leq k \E{(S-\alpha)^+}, 
\]
where we have used the fact that $\E{S} = \rho \E{A}$.
Consider the change of $V$ between two consecutive arrivals:
\[
    \E{V(A)} - V(0) &= \int_0^A \dot{V}(t) \d{t} +  \lambda \Delta V \\
    &\leq \E[\Big]{\int_0^A  \gp[\big]{-(1-\rho) + (1-J(t))\indic{\Sbig(t) \leq \alpha}}  \d{t} } +  \lambda k \E[\Big]{(S-\alpha)^+}  \\
    &=  - \lambda (1-\rho) +  \E[\Big]{\int_0^A  (1-J(t))\indic{\Sbig(t) \leq \alpha}  \d{t} }+ \lambda k \E[\Big]{(S-\alpha)^+}.
\]
Let $\mathcal{R} = [0, k(\alpha+\beta)+\beta]\times [0, \alpha+\beta]$, a compact region. We show that $\E{V(A)} - V(0) < 0$ as long as $(W(0), \Sbig(0)) \notin \mathcal{R}$. There are two cases to consider:
\* If $\Sbig(0) > \alpha + \beta$, then $\Sbig(t) > \alpha + \beta - t$. As a result,
\[
    \E{V(A)} - V(0) &\leq - \lambda (1-\rho) + \E{(A - \beta)^+} + \lambda k \E{(S-\alpha)^+}. 
\]
\* If $\Sbig(0) \leq \alpha + \beta$ but $W(0) > k (\alpha+\beta)+\beta$, then we have $\Sbig(t) \leq \alpha+\beta$, and $W(t) \geq k(\alpha+\beta) + \beta - t$. For $t \leq \beta$, $N(t) \geq W(t)/\Sbig \geq k$, implying $J(t) = 1$. As a result,
\[
    \E{V(A)} - V(0) \leq - \lambda (1-\rho) + \E{(A-\beta)^+} + \lambda k \E{(S-\alpha)^+}. 
\]
\*/
Taking sufficiently large $\alpha$ and $\beta$, we get $\E{V(A)} - V(0) < 0$ in both cases. 

\subsection{Remaining Obstacles}

We have given a Lyapunov function that has negative drift for the embedded Markov process when $(W, \Sbig)$ is outside a compact region $\mathcal{R}$. However, $(W, \Sbig)$ is not a complete state description of the Markov process. Moreover, the set of states such that $(W, \Sbig)\in\mathcal{R}$ might not be compact. We believe one can show that this set of states is \emph{petite} \citep{meyn_markov_2009}, even if it is not compact, which will suffice. However, this seems to require reasoning about the details of the job state space.

Of course, so far we have just been focusing on the existence of a stationary distribution for the embedded process at arrival instants. Once this is established, we need to use it to construct a stationary distribution for the continuous-time process, but this step is more standard.